\documentclass[aps,twocolumn,superscriptaddress,dvipsnames,nofootinbib,amssymb,floatfix,10pt,prx]{revtex4-2} 
\usepackage{graphicx}
\usepackage{dcolumn}
\usepackage{bm}
\usepackage{qcircuit}
\usepackage{braket}
\usepackage{color}
\usepackage{amsthm}
\usepackage{amsmath}
\usepackage{chemarrow}
\usepackage{overpic}
\usepackage{diagbox}
\usepackage{braket}
\usepackage{subfigure}
\usepackage{makecell}
\usepackage{tikz}
\newtheorem{theorem}{Theorem}
\newtheorem{definition}{Definition}
\newtheorem{lemma}{Lemma}
\newtheorem{corollary}{Corollary}
\usepackage{chngcntr}
\usepackage[colorlinks=true,linkcolor=red,bookmarks=true,breaklinks=true,citecolor=blue]{hyperref}
\usepackage{cleveref}
\let\oldacl\addcontentsline
\renewcommand{\addcontentsline}[3]{}
\newcommand{\Tr}{\mathrm{Tr}}

\renewcommand{\d}{\mathrm{d}}

\definecolor{NG}{rgb}{0.66,0.29,0.48}

\newcommand{\addressCHEM}{Department of Chemistry, University of Toronto, Lash Miller Chemical Laboratories, 80 St. George Street, ON M5S 3H6, Toronto, Canada}
\newcommand{\addressAC}{Acceleration Consortium, 700 University Ave., M7A 2S4, Toronto, Canada}
\newcommand{\addressCS}{Department of Computer Science, University of Toronto, Sandford Fleming Building, 10 King’s College Road, ON M5S 3G4, Toronto, Canada}
\newcommand{\addressVECTOR}{Vector Institute for Artificial Intelligence, 661 University Ave. Suite 710, ON M5G 1M1, Toronto, Canada}
\newcommand{\addressMSE}{Department of Materials Science \& Engineering, University of Toronto, 184 College St., M5S 3E4, Toronto, Canada}
\newcommand{\addressCHEMENG}{Department of Chemical Engineering \& Applied Chemistry, University of Toronto, 200 College St. ON M5S 3E5, Toronto, Canada}
\newcommand{\addressCIFAR}{Senior Fellow, Canadian Institute for Advanced Research (CIFAR), 661 University Ave., M5G 1M1, Toronto, Canada}
\newcommand{\addressNVIDIA}{NVIDIA, 431 King St W \#6th, M5V 1K4, Toronto, Canada}

\begin{document}
\title{Fast-forwardable Lindbladians imply quantum phase estimation}
\author{Zhong-Xia Shang}
\email{shangzx@hku.hk}
\affiliation{HK Institute of Quantum Science $\&$ Technology, The University of Hong Kong, Hong Kong, China}
\affiliation{QICI Quantum Information and Computation Initiative, Department of Computer Science,
The University of Hong Kong, Hong Kong, China}

\author{Naixu Guo}
\affiliation{Centre for Quantum Technologies, National University of Singapore, 117543, Singapore}
\author{Patrick Rebentrost}
\affiliation{Centre for Quantum Technologies, National University of Singapore, 117543, Singapore}
\affiliation{School of Computing, National University of Singapore, 117417, Singapore}
\author{Al\'an Aspuru-Guzik}
\affiliation{\addressCHEM}
\affiliation{\addressAC}
\affiliation{\addressCS}
\affiliation{\addressVECTOR}
\affiliation{\addressMSE}
\affiliation{\addressCHEMENG}
\affiliation{\addressCIFAR}
\affiliation{\addressNVIDIA}
\author{Tongyang Li}
 \email{tongyangli@pku.edu.cn}
\affiliation{Center on Frontiers of Computing Studies, Peking University, Beijing 100871, China}
\affiliation{School of Computer Science, Peking University, Beijing 100871, China}
\author{Qi Zhao}
\thanks{\href{mailto:zhaoqcs@hku.hk}{zhaoqcs@hku.hk}\\Z.-X. S. and N. G. contributed equally.}
\affiliation{QICI Quantum Information and Computation Initiative, Department of Computer Science,
The University of Hong Kong, Hong Kong, China}

\begin{abstract}

Quantum phase estimation (QPE) and Lindbladian dynamics are both foundational in quantum information science and central to quantum algorithm design. In this work, we bridge these two concepts: certain simple Lindbladian processes can be adapted to perform QPE-type tasks. However, unlike QPE, which achieves Heisenberg-limit scaling, these Lindbladian evolutions are restricted to standard quantum limit complexity. This indicates that, different from Hamiltonian dynamics, the natural dissipative evolution speed of such Lindbladians does not saturate the fundamental quantum limit, thereby suggesting the potential for quadratic fast-forwarding. We confirm this by presenting a quantum algorithm that simulates these Lindbladians for time $t$ within an error $\varepsilon$ using $\mathcal{O}\left(\sqrt{t\log(\varepsilon^{-1})}\right)$ cost, whose mechanism is fundamentally different from the fast-forwarding examples of Hamiltonian dynamics. As a bonus, this fast-forwarded simulation naturally serves as a new Heisenberg-limit QPE algorithm. Therefore, our work explicitly bridges the standard quantum limit-Heisenberg limit transition to the fast-forwarding of dissipative dynamics. We also adopt our fast-forwarding algorithm for efficient Gibbs state preparation and demonstrate the counter-intuitive implication: the allowance of a quadratically accelerated decoherence effect under arbitrary Pauli noise.



\end{abstract}
\maketitle
\noindent\textit{\textbf{Introduction.—}}Quantum algorithms leverage the properties of quantum mechanics to achieve acceleration over classical rivals \cite{nielsen2010quantum}. Quantum phase estimation (QPE) \cite{nielsen2010quantum,kitaev1995quantum} is a fundamental primitive at the core of many quantum algorithms \cite{dalzell2023quantum}, enabling the estimation of eigenvalues and the preparation of eigenstates for unitary operators. QPE and its various variants \cite{wiebe2016efficient,ding2023even,higgins2007entanglement} serve as the foundation for numerous significant quantum applications, such as factoring large numbers \cite{shor1994algorithms}, solving linear systems of equations \cite{harrow2009quantum}, and estimating ground state energies of quantum Hamiltonians in condensed matter physics and quantum chemistry \cite{mcardle2020quantum}.

Simulating quantum dynamics \cite{feynman2018simulating} is one of the most important applications for quantum computing. QPE has deep connections with Hamiltonian dynamics, which can estimate eigenvalues and prepare eigenstates of a Hamiltonian $H$ by querying Hamiltonian dynamics $e^{-iHt}$. To estimate eigenvalues within an error $\epsilon$, the required total Hamiltonian evolution time is $\mathcal{O}(\epsilon^{-1})$, which is known to achieve the Heisenberg limit \cite{zwierz2010general}, a fundamental limit set by the energy-time uncertainty principle of quantum mechanics \cite{busch2007heisenberg,kijowski1974time}. This fundamental limit implies the Hamiltonian no fast-forwarding theorem \cite{atia2017fast,berry2007efficient}, which explains why existing universal Hamiltonian simulation algorithms \cite{lloyd1996universal,berry2015simulating,low2017optimal} generally have $\Omega(t)$ complexity scalings.
 

Quantum dynamics is not restricted to closed systems. Lindbladians \cite{lindblad1976generators} describe the dissipative dynamics of open quantum systems and havee recently been shown to be a transformative quantum algorithm primitive for various critical tasks, including preparations for the Gibbs states \cite{chen2023quantum, chen2023efficient, rouze2024optimal} and ground state \cite{ding2024single}, optimizations \cite{chen2025quantum}, optimal control \cite{he2024efficient}, and solving differential equations \cite{shang2024design}. Given the importance of Lindbladian simulations \cite{kliesch2011dissipative,cleve2017efficient,li2022simulating, ding2025lower} and the deep connections between QPE and Hamiltonian dynamics, it is therefore natural and important to explore the connections between QPE and Lindbladians and their implications for the Lindbladian fast-forwarding. Moreover, as quantum algorithms based on QPE and Lindbladians seem to be rather different, understanding their connections will also be vital for bridging unitary and dissipative quantum techniques and recognizing the strengths and weaknesses of these two primitives compared to each other. 

In this work, we discover interesting and non-trivial connections between QPE and Lindbladians. Specifically, we demonstrate that there exist simple, purely dissipative Lindbladians whose simulation can be utilized for QPE tasks, estimating eigenvalues and preparing eigenstates of Hamiltonians. However, this Lindbladian-based QPE only achieves a standard quantum limit scaling $\mathcal{O}(\epsilon^{-2})$ and thus, has not yet reached the minimal energy-time uncertainty. In other words, the natural dynamics of these open quantum systems do not fully exploit quantum advantages with respect to the estimation accuracy. This observation motivates the prospect of quadratic fast-forwarding for such Lindbladian dynamics. 

Indeed, we propose a quantum algorithm that can simulate these Lindbladians of evolution time $\mathcal{O}(t)$ with a complexity of only $\mathcal{O}(\sqrt{t})$. Compared with the commonly used dilated Hamiltonian approach (see explanations in Ref.~\cite{cleve2017efficient}) for Lindbladian simulation, our algorithm also has exponential advantages in terms of accuracy and number of ancilla qubits. While in the context of Hamiltonian dynamics, several special examples that enable fast forwarding have been found \cite{atia2017fast,gu2021fast}, to the best of our knowledge, this work gives the first example of Lindbladian fast forwarding. Notably, the reason for Lindbladian fast forwarding here is the non-reachability of the Heisenberg limit, which is in stark contrast to Hamiltonian cases, whose reason is that their special structures enable overcoming the Heisenberg limit \cite{atia2017fast}.

Since our algorithm realizes quadratically fast-forwarding for Lindbladians that enable QPE tasks, it directly leads to a Lindbladian-based Heisenberg-limit ($\mathcal{O}(\epsilon^{-1})$) QPE. Therefore, we successfully bridge the standard quantum limit-Heisenberg limit transition to the fast-forwarding of dissipative dynamics. More interestingly, the fast-forwarding algorithm also unveils a surprising understanding of quadratic quantum speedup that seems to have nothing to do with quantum. The core of our fast-forwarding simulation algorithm is that the expected translation distance of $N$-step ($t$) \textbf{classical} random walk \cite{xia2019random} is $\mathcal{O}(\sqrt{N})$ ($\mathcal{O}\left(\sqrt{t}\right)$). Since QPE accounts for the Grover-type~\cite{grover1996fast} quantum advantage through amplitude estimation \cite{brassard2000quantum}, thus, the $\mathcal{O}(\sqrt{N})$ spreading of the \textbf{classical} random walk becomes exactly the source of quadratic quantum speedup. 

Besides these main results, we also show that our fast-forwarding simulation algorithm can be applied to quantum Gibbs state preparations, which can match the state-of-the-art scaling in quantum singular value transformation \cite{gilyen2019quantum} approach. Moreover, we find that our fast-forwarding simulation algorithm also works for a much broader class of Lindbladians, which we call the Choi commuting Lindbladians, that include arbitrary combinations of Pauli noise. This indicates that the decoherence effect under arbitrary Pauli noise can actually occur quadratically faster.

We want to mention here another work of Lindbladian fast-forwarding~\cite{shang2025exponential}. The fast-forwarding there is only for circuit depth, while keeping the query complexity or the total cost unchanged through increasing the number of ancilla qubits. In contrast, in this work, the Lindbladian fast-forwarding considers the total cost.

\noindent\textit{\textbf{Preliminaries.—}}In this work, we express an $n$-qubit Hamiltonian in the form
\begin{eqnarray}\label{mainhami}
H=\sum_\alpha h_\alpha \Pi_\alpha,
\end{eqnarray}
where $\{h_\alpha\}\in[0,1]$ are real eigenvalues satisfying $h_\alpha\neq h_\beta,~\forall \alpha\neq \beta$ and $\{\Pi_\alpha\}$ are the corresponding eigenspace projectors. We define two tasks.
\begin{definition}[Eigenvalue estimation]\label{task1}
Input $|\psi_\alpha\rangle$ that is an eigenstate of $H$ with eigenvalue $h_\alpha$, output an estimate of $h_\alpha$ to an additive error $\epsilon$ with a success probability at least $1-\delta$. 
\end{definition}
\begin{definition}[Eigenstate preparation]\label{task2}
Input $|\psi\rangle=\sum_\alpha \Pi_\alpha |\psi\rangle=\sum_\alpha c_\alpha |\psi_\alpha\rangle$ and the knowledge of $h_\beta$, output a quantum state $|\tilde{\psi}_\beta\rangle$ that has the overlap with the eigenstate $|\psi_\beta\rangle$ at least $1-\zeta$: $|\langle \tilde{\psi}_\beta |\psi_\beta\rangle|^2\geq 1-\zeta$.
\end{definition}

QPE is a quantum algorithm that can solve these two tasks. The standard QPE \cite{nielsen2010quantum} has the circuit shown in Fig. \ref{fig.qpe} in SM \ref{qpeintro}, which adopts the quantum Fourier transform \cite{shor1994algorithms}. For $|\psi\rangle=\sum_\alpha c_\alpha |\psi_\alpha\rangle$, as we gradually increase the number of ancilla qubits, the output state of QPE will approach $\sum_\alpha c_\alpha|h_\alpha\rangle|\psi_\alpha\rangle$, where $|h_\alpha\rangle$ is a computational basis state corresponding to the binary representation of $h_\alpha$. Therefore, measuring the ancilla qubits will tell us the eigenvalues of $H$, and the resulting system state will collapse to the corresponding eigenstate. For the eigenvalue estimate task (setting $|\psi\rangle=|\psi_\alpha\rangle$), standard QPE will require a Hamiltonian simulation time
\begin{eqnarray}\label{ee}
\text{Hami.simu.time}=
\mathcal{O}(\epsilon^{-1}\delta^{-1}),
\end{eqnarray}
and for the eigenstate preparation task, the required Hamiltonian simulation time with the aid of amplitude amplification \cite{brassard2000quantum} is
\begin{eqnarray}\label{ep}
\text{Hami.simu.time}=\mathcal{O}(|c_\beta|^{-2}\Delta_\beta^{-1}\zeta^{-1/2}),
\end{eqnarray}
where $\Delta_\beta$ is the gap between $h_\beta$ and other eigenvalues (See SM \ref{qpeintro} for more details). Note that, in this work, we will mainly use the Hamiltonian simulation time as a measure of complexity. However, one can always translate this into the actual gate complexity on digital quantum computers by further considering various Hamiltonian simulation algorithms \cite{lloyd1996universal,berry2015simulating,low2017optimal}. We want to mention that the complexity in Eq.~(\ref{ee}) and Eq.~(\ref{ep}) is for the standard QPE, and the dependence on $\delta$, $c_\beta$, and $\zeta$ parameters can be further improved through more advanced techniques \cite{hozo2005estimating,childs2010relationship,lin2020near,lin2022heisenberg}, including our method shown in this work.

The $\epsilon^{-1}$ scaling is known as the Heisenberg limit, which reaches the minimal quantum energy-time uncertainty principle and therefore is the fundamental lower bound of eigenvalue estimation \cite{zwierz2010general}. This result also leads to the no-fast-forwarding theorem for Hamiltonian simulation \cite{berry2007efficient,atia2017fast}, since if there exists a quantum algorithm that can simulate $e^{-iHt}$ in a time $\mathit{o}(t)$ for any $H$, the Heisenberg limit will be broken. Nonetheless, there do exist special cases of Hamiltonians whose special structures enable surpassing the Heisenberg limit and realizing the fast-forwarding, such as Hamiltonians from the factoring algorithm \cite{shor1994algorithms}, commuting local Hamiltonians, and quadratic fermionic Hamiltonians \cite{atia2017fast,gu2021fast}. 

In contrast with Hamiltonian dynamics (Schr{\"o}dinger equation) describing closed quantum systems, the dynamics of open quantum systems with weak system-environment interactions can be described by Lindbladian dynamics \cite{lindblad1976generators}, which takes the form
\begin{eqnarray}\label{mainlme}
&&\frac{\d\rho(t)}{\d t}=\mathcal{L}[\rho(t)]\\&&=-i[H_I,\rho(t)]+
\sum_i \left(F_i\rho(t) F_i^\dag-
\frac{1}{2}\{\rho(t),F_i^\dag F_i\}\right),\nonumber
\end{eqnarray}
where $\rho(t)$ is the system density matrix, $H_I$ is the internal Hamiltonian, and $F_i$ are quantum jump operators describing dissipative interactions between the system and the environment. While the Lindbladian is no longer a unitary operation, it is still a completely positive and trace-preserving (CPTP) map \cite{nielsen2010quantum}. The Lindbladian in Eq.~(\ref{mainlme}) can be simulated through the dilated Hamiltonian simulation approach (see SM \ref{smc}). For example, consider a Lindbladian with $H_I=0$, and there is only a single jump operator
\begin{equation}\label{mainflme}
\frac{\d\rho(t)}{\d t}=\mathcal{L}_F[\rho(t)]=F\rho(t) F^\dag-
\frac{1}{2}\{\rho(t),F^\dag F\},
\end{equation}
then the short-time dynamics of the dilated Hamiltonian $\tilde{F}=\begin{pmatrix}
0 & F^\dag \\
F & 0
\end{pmatrix}$ will approximate the short-time Lindbladian dynamics
\begin{eqnarray}
\Tr_a\left(e^{-i \tilde{F}\sqrt{\tau}}(|0\rangle\langle 0|\otimes \rho(t))e^{i \tilde{F}\sqrt{\tau}}\right)\approx \rho(t+\tau).
\end{eqnarray}
Therefore, by repeatedly adding ancilla qubits and running short-time $\tilde{F}$ dynamics, we are able to approximately simulate Eq.~(\ref{mainflme}). For complexity, to simulate the Lindbladian Eq.~(\ref{mainflme}) for time $t$ within error $\varepsilon$, the required dilated Hamiltonian simulation time and number of ancilla qubits are
\begin{eqnarray}
\text{Hami.simu.time}\,&&=\mathcal{O}(t^{2}\varepsilon^{-1}),\label{dl1}\\
\text{Number of ancilla qubits}\,&&=\mathcal{O}(t^3\varepsilon^{-2}).\label{dl2}
\end{eqnarray}
Note that more advanced Lindbladian simulation algorithms \cite{cleve2017efficient,li2022simulating} can achieve $\tilde{\mathcal{O}}(t\text{poly}\log(\varepsilon^{-1}))$ and Ref.~\cite{childs2016efficient} also gives the no fast-forwarding theorem for general Lindbladians.

\noindent\textit{\textbf{Lindbladian as slow QPE.—}}At first glance, coherent QPE and dissipative Lindbladians seem to be unrelated. However, our key observation is that if we set the jump operator $F$ in Eq.~(\ref{mainflme}) to be exactly the Hamiltonian $H$, the Lindbladian
\begin{equation}\label{mainslme}
\frac{\d \rho(t)}{\d t}=\mathcal{L}_H[\rho(t)]=H\rho(t) H-
\frac{1}{2}\{\rho(t),H^2\}
\end{equation}
will let the system experience the dephasing effect on the eigenbasis of $H$ (SM \ref{smmoti}). More concretely, if $\rho(0)=|\psi\rangle\langle\psi|$, the steady state of Eq.~(\ref{mainslme}) will be $\rho_{\text{ss}}=\sum_\alpha|c_\alpha|^2|\psi_\alpha\rangle\langle\psi_\alpha|$. The point here is that the steady state is exactly the system reduced density matrix for $\sum_\alpha c_\alpha|h_\alpha\rangle|\psi_\alpha\rangle$ in QPE under infinite accuracy (infinite ancillas). Therefore, as both QPE and the Lindbladian in Eq.~(\ref{mainslme}) lead to dephasing effects on the eigenbasis of $H$, we can expect that the simulation of Eq.~(\ref{mainslme}) can also be applied to QPE tasks (Definition \ref{task1}-\ref{task2}).

\begin{figure}[htbp]
\centering
\includegraphics[width=0.49\textwidth]{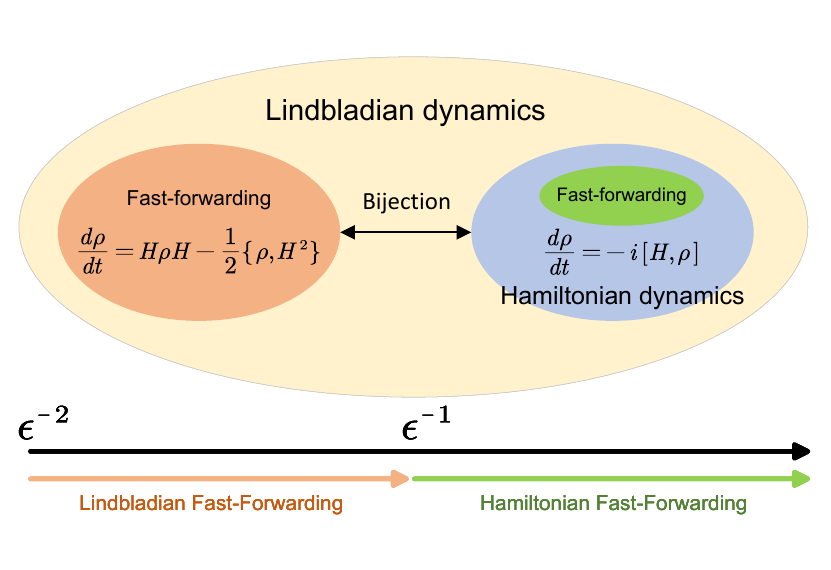}
\caption{Comparison between Lindbladian fast-forwarding and Hamiltonian fast-forwarding. $H$ can be put in either Hamiltonian dynamics or the Lindbladian Eq.~(\ref{mainslme}). Since Hamiltonian dynamics naturally achieves the Heisenberg limit, we can only expect the fast-forwarding for quite restricted Hamiltonians (green) whose special properties allow for surpassing the Heisenberg limit. In contrast, the natural evolution of the Lindbladian Eq.~(\ref{mainslme}) only achieves the standard quantum limit, meaning improving to the Heisenberg limit is sufficient for fast-forwarding, which works for an arbitrary Hamiltonian (orange).\label{fig2}}
\end{figure}

Indeed, for the dilated Hamiltonian simulation ($\tilde{H}=\begin{pmatrix}
0 & H \\
H & 0
\end{pmatrix}$) of Eq.~(\ref{mainslme}), we have 
\begin{eqnarray}\label{mainshort}
&&e^{-i \tilde{H}\sqrt{\tau}}|0\rangle|\psi\rangle=\\&&|0\rangle\cos(\sqrt{\tau}H)|\psi\rangle-i|1\rangle\sin(\sqrt{\tau}H)|\psi\rangle.\nonumber
\end{eqnarray}
If we measure the ancillas after implementing Eq.~(\ref{mainshort}) for $N$ times ($N$ ancilla qubits), and the measurement outcome on the computational basis of ancilla qubits has $N-m$ 0s and $m$ 1s, then the collapsed system state will concentrate on the eigensubspace of $H$ with eigenvalue $\sqrt{N/t}\arcsin(\sqrt{m/N})$. Based on this, we can design a Lindbladian-based QPE (SM \ref{smslow}) with the complexity summarized in the following theorem.
\begin{theorem}\label{the1}
Through the dilated Hamiltonian approach for simulating the Lindbladian Eq.~(\ref{mainslme}), for QPE tasks, the required Lindbladian evolution time and dilated Hamiltonian ($\tilde{H}$) simulation time are
\begin{eqnarray}
\begin{tabular}{|c|c|c|}
\hline
 & Eigenvalue &Eigenstate \\
\hline
Lind& $\mathcal{O}(\epsilon^{-2}\log(\delta^{-1}))$ & $\mathcal{O}(|c_\beta|^{-1}\Delta_\beta^{-2}\log(\zeta^{-1}))$\\
\hline
Hami ($\tilde{H}$) & $\mathcal{O}(\epsilon^{-4}\log^2(\delta^{-1}))$ & $\mathcal{O}(|c_\beta|^{-1}\Delta_\beta^{-4}\log^2(\zeta^{-1}))$\\\hline
\end{tabular}.\nonumber
\end{eqnarray}
\end{theorem}
Compared with Eq.~(\ref{ee}) and Eq.~(\ref{ep}) in standard QPE, we can find that the results in Theorem \ref{the1} have exponential advantages in terms of $\delta$ and $\zeta$, and quadratic advantage in terms of $|c_\beta|$ (here we also use the amplitude amplification \cite{brassard2000quantum} technique). However, the factors $\epsilon$ and $\Delta_\beta$ are quadratically worse in the Lindbladian evolution time and even quartically worse in the dilated Hamiltonian simulation time, which is the real complexity when running on quantum computers. This $\epsilon^{-2}$ scaling is known as the standard quantum limit and indicates that the natural evolution of the Lindbladian Eq.~(\ref{mainslme}) has not reached the minimal energy-time uncertainty, and therefore is not that ``quantum" compared with the Hamiltonian dynamics.

\noindent\textit{\textbf{Quadratic fast-forwarding simulation.—}}Since the natural evolution of the Lindbladian Eq.~(\ref{mainslme}) has not reached the Heisenberg limit, the energy-time uncertainty principle does not prohibit its fast-forwarding, and the existence of a Lindbladian simulation algorithm for Eq.~(\ref{mainslme}) with complexity  $\mathcal{O}(\sqrt{t})$ is expected. Indeed, we can observe that Eq.~(\ref{mainshort}) satisfies
\begin{eqnarray}
&&(\textbf{H}\otimes I)e^{-i \tilde{H}\sqrt{\tau}}|0\rangle|\psi\rangle=\nonumber\\&&\frac{1}{\sqrt{2}}|0\rangle e^{-iH\sqrt{\tau}}|\psi\rangle+\frac{1}{\sqrt{2}}|1\rangle e^{iH\sqrt{\tau}}|\psi\rangle,
\end{eqnarray}
which acts as a $|0\rangle$/$|1\rangle$-controlled forward and backward $H$ evolution. Since each $e^{- i\tilde {H}\sqrt{\tau}}$ uses a new ancilla qubit, the net behaviors of implementing the short-time dilated Hamiltonian $\tilde{H}$ simulation for $N$ times can be well captured by an $N$-step 1-dimensional classical random walk with equal probability \cite{xia2019random}. As a result, while the maximal Hamiltonian simulation time is $N\sqrt{\tau}$, the majority will concentrate in the simulation time range $\sqrt{N\tau}$. Based on this observation, we give a Lindbladian simulation algorithm for Eq.~(\ref{mainslme}) with the basic quantum circuit shown in Fig. \ref{fig.fflind} of SM \ref{smff} and the complexity summarized in Theorem \ref{the2}.
\begin{theorem}\label{the2}
The Lindbladian Eq.~(\ref{mainslme}) for an evolution time $t$ can be simulated in $\varepsilon$ error with a complexity
\begin{eqnarray}
\text{Hami.simu.time ($H$)}&&=\mathcal{O}\left(t^{1/2}\log^{1/2}(\varepsilon^{-1})\right),\nonumber\\
\text{Number of ancilla qubits}&&=\mathcal{O}\left(\log(t)+\log(\varepsilon^{-1})\right).\nonumber
\end{eqnarray}
\end{theorem}
Our results present a quartic advantage in $t$ compared to that of the dilated Hamiltonian approach (Eqs.~(\ref{dl1})-(\ref{dl2})), corresponding to a quadratic fast forwarding of the natural Lindbladian evolution. Moreover, it also achieves exponential advantages in terms of $\epsilon$ and the number of ancilla qubits. Note that, in our fast-forwarding algorithm, we need the preparation of the binomial distribution state $2^{-N/2}\sum_{m=0}^N\sqrt{\binom{N}{m}}|m\rangle$ with $|m\rangle$ in binary representation. In SM \ref{sec.stateprep}, based on Ref.~\cite{kitaev2008wavefunction}, we show that $\mathcal{O}(\text{poly}(n+\log(\varepsilon^{-1})))$ gate complexity is sufficient to prepare the state within an error $\varepsilon$.

Our Lindbladian fast-forwarding unveils non-trivial differences between unitary and dissipative dynamics (Fig. \ref{fig2}). For Hamiltonian dynamics, since they already meet the Heisenberg limit, one can only expect rather restrictive cases with special properties, such as being efficiently diagonalizable \cite{atia2017fast}, to overcome the Heisenberg limit and therefore, enable fast-forwarding. In contrast, the Lindbladian Eq.~(\ref{mainslme}) only achieves the standard quantum limit, and one only needs to let it reach the Heisenberg limit (see the following section) to enable fast-forwarding. Also, while Eq.~(\ref{mainslme}) is restrictive in terms of the Lindbladian, its cardinality is actually comparable to the whole Hamiltonian dynamics. For any $H$, we can either put it in the Hamiltonian dynamics $e^{-iHt}$ or the dissipative Lindbladian Eq.~(\ref{mainslme}), indicating they form a bijection.

\noindent\textit{\textbf{Fast forwarding for fast QPE.—}}In Theorem \ref{the1}, we show that the Lindbladian Eq.~(\ref{mainslme}) is related to the QPE tasks, and in Theorem \ref{the2}, we show that Eq.~(\ref{mainslme}) can be quadratically fast-forwarded, therefore, naturally, the combination of these results should lead to a Heisenberg-limit Lindbladian-based QPE algorithm (SM \ref{smffqpe}). We summarize the complexity of this new QPE algorithm in the following theorem
\begin{theorem}\label{the3}
The Lindbladian fast-forwarding algorithm (Theorem \ref{the2}) can be adopted for QPE tasks, the required Hamiltonian simulation time ($H$) is
\begin{eqnarray}
\begin{tabular}{|c|c|c|}
\hline
 & Eigenvalue &Eigenstate \\
\hline
Hami ($H$) & $\mathcal{O}(\epsilon^{-1}\log(\delta^{-1}))$ & $\mathcal{O}(|c_\beta|^{-1}\Delta_\beta^{-1}\log(\zeta^{-1}))$\\\hline
\end{tabular}.\nonumber
\end{eqnarray}
\end{theorem}

More interestingly, Theorem \ref{the3} provides a new way of emerging quadratic quantum speedup, which seems to have nothing to do with quantum. The point is that the Lindbladian-based QPE in Theorem \ref{the3} can be used for amplitude estimation \cite{brassard2000quantum}. By encoding the number of solutions of NP-complete decision problems into target amplitudes, one can use amplitude estimation to solve the decision problems \cite{aaronson2020quantum,brassard2000quantum}, and the Heisenberg limit scaling of QPE leads to the quadratic quantum speedup for NP-complete problems (See SM \ref{smqua}). Recall that the Heisenberg limit in Theorem \ref{the3} is due to our fast-forwarding algorithm, Theorem \ref{the2}, whose core is merely the concentration of \textbf{classical} random walk! In other words, the behaviors of classical random walk exactly provide the source of quadratic quantum speedup. This is in stark contrast with the quadratic speedup of hitting time in quantum walk \cite{venegas2012quantum}, where the ability of quantum walk to go beyond the behaviors of classical random walk provides the source.

\noindent\textit{\textbf{Additional results.—}}We provide two additional results.

\noindent\textit{Quantum Gibbs state preparation: }It turns out that our Lindbladian fast-forwarding algorithm can also be used for quantum Gibbs state preparation. Given a positive semidefinite Hamiltonian $H_P$ with $\|H_P\|\leq 1$, our goal is to prepare the Gibbs state $\rho_\beta=e^{-\beta H_P}/Z_\beta$ at the inverse temperature $\beta$ with $Z_\beta=\Tr(e^{-\beta H_P})$ the partition function. Based on Ref.~\cite{shang2024design}, in SM \ref{smpar}, we show that by setting $F$ in Eq.~(\ref{mainflme}) as $F=\begin{pmatrix}
\sqrt{H_P}\otimes I &0\\
0 & 0\end{pmatrix}$, one can eventualy prepare the purification $|\rho_\beta\rangle$ of $\rho_\beta$. Since $F$ is Hermitian, which satisfies our fast-forwarding requirements, as a result, through the state-of-the-art Hamiltonian simulation algorithm \cite{low2017optimal}, to prepare $|\rho_\beta\rangle$ within an error $\varepsilon$, the query complexity on the block encoding of $\sqrt{H_P}$ is
\begin{eqnarray}
\tilde{\mathcal{O}}\left(\sqrt{\frac{2^n}{Z_\beta}\beta\log(\varepsilon^{-1})} \right),
\end{eqnarray}
for $\beta> \log^3(\varepsilon^{-1})$. This result matches the previous state-of-the-art results shown in QSVT~\cite{gilyen2019quantum} at large $\beta$, and has better scaling in $\varepsilon$ compared with results in Ref.~\cite{chowdhury2016quantum} (Both of the methods use $\sqrt{H_P}$ access). It is important to note that the $\sqrt{\beta}$ scaling in our approach is due to our fast-forwarding algorithm, whose origin is the concentration of the classical random walk and is in stark contrast to QSVT, where the $\sqrt{\beta}$ scaling is due to the fast polynomial approximations of monomial functions~\cite{sachdeva2014faster}. We also want to mention that there is a series of works \cite{chen2023quantum,chen2023efficient,rouze2024optimal} for quantum Gibbs state preparation through Davies generator-type Lindbladians recently, which maps the Gibbs states into the steady state of Lindbladians. Since their complexity mainly relies on the mixing time of the Lindbladian, which can vary greatly from case to case, we currently don't know how to make a fair comparison with those results.

\noindent\textit{Choi commuting Lindbladians and Pauli noise: }In fact, our quadratic fast-forwarding algorithm can be applied to a much broader class of Lindbladians beyond Eq.~(\ref{mainslme}). 
For the following Lindbladian
\begin{equation}
\frac{\d\rho(t)}{\d t}=\sum_{i}\left(H_i\rho(t) H_i^\dag-
\frac{1}{2}\{\rho(t),H_i^\dag H_i\}\right),
\end{equation}
as long as $H_i$ are Hermitian and 
\begin{eqnarray}\label{condd}
&&[H_i\otimes H_i^*-\frac{1}{2}H_i^2\otimes I-I\otimes \frac{1}{2}H_i^{*2},\nonumber\\&&H_j\otimes H_j^*-\frac{1}{2}H_j^2\otimes I-I\otimes \frac{1}{2}H_j^{*2}]=0,
\end{eqnarray}
is satisfied for any $i,j$, we allow to achieve a fast-forwarding simulation with the required Hamiltonian simulation time also of order $\mathcal{O}\left(t^{1/2}\log^{1/2}(\varepsilon^{-1})\right)$. We name such Lindbladians the Choi commuting Lindbladians (See SM \ref{smgen}). Importantly, Eq.~(\ref{condd}) can be satisfied when $H_i$ and $H_j$ are either commuting or anti-commuting. The multi-qubit Pauli group exactly satisfies this condition. Therefore, we obtain the surprising result that for any combination of Pauli noises, the resulting decoherence effect can actually happen in a quadratically faster speed. Since Pauli noise is the standard error model in the quantum error correction community~\cite{lidar2013quantum}, our results suggest the potential applications of using quantum computers for exploring noisy quantum systems. Besides Pauli noise, such Choi commuting Lindbladians also cover other important classes such as maximal dephasing quantum master equations~\cite{fagnola2019mathematical}.

\noindent\textit{\textbf{Summary and outlook.—}}In summary, we find a deep connection between QPE and Lindbladians, and its implications to fast-forwarding. Theorems \ref{the1}-\ref{the3} are logically self-consistent. Lindbladian Eq.~(\ref{mainslme}) for QPE only gives a standard quantum limit scaling (Theorem \ref{the1}), which indicates the possibility of quadratic fast-forwarding (Theorem \ref{the2}). The fast-forwarding algorithm, in turn, improves QPE based on Eq.~(\ref{mainslme}) to the Heisenberg limit (Theorem \ref{the3}). We unveil the stark differences of fast-forwarding in Hamiltonian and Lindbladian dynamics and give a new way of achieving quadratic quantum speedup stemming from the concentration of a classical random walk. We also show how to utilize the Lindbladian fast-forwarding for quantum Gibbs state preparations and its generalizations to Choi commuting Lindbladians, including arbitrary combinations of Pauli noise.

There are several interesting future directions. First, we can ask whether there are other types of Lindbladian fast-forwarding that can go beyond quadratic to exponential based on different mechanisms. Second, in SM~\ref{smffqpe}, we show that our Lindbladian fast QPE algorithm has something to do with the Kravchuk polynomial~\cite{krawtchouk1929generalisation}, which might draw special interest. Third, since the Lindbladian Eq.~(\ref{mainslme}) we considered can help prepare diagonal quantum states on the eigenbasis of $H$, it naturally catalyzes the quantum Mpemba effect~\cite{moroder2024thermodynamics}: an exponential speedup towards the thermal equalibrism can be achieved for diagonal input states compared to states with coherence on the eigenbasis of $H$. Further studies in this direction could lead to improvements on recent Lindbladian-based quantum Gibbs state preparation algorithms~\cite{chen2023quantum,chen2023efficient,rouze2024optimal}. Forth, we hope that this work can inspire more research on the comparison between dissipative and coherent quantum algorithms, especially their advantages and disadvantages in different tasks. Finally, it is exciting to explore more quantum algorithms based on Lindbladians and open quantum systems, where new types of quantum speedup might be discovered.
$$\\$$
\begin{acknowledgments}
The authors would like to thank Itai Arad, Mohsen Bagherimehrab, Richard Cleve, Hsin-Yuan Huang, Guang Hao Low, Peixue Wu, Yuxiang Yang, and Xiao-Ming Zhang for valuable discussions. Z.S. would also like to thank Dong An and Changpeng Shao for discussions on a related project. Z.S. and Q.Z. acknowledge the support from Innovation Program for Quantum Science and Technology via Project 2024ZD0301900, National Natural Science Foundation of China (NSFC) via Project No. 12347104 and No. 12305030, Guangdong Basic and Applied Basic Research Foundation via Project 2023A1515012185, Hong Kong Research Grant Council (RGC) via No. 27300823, N\_HKU718/23, and R6010-23, Guangdong Provincial Quantum Science Strategic Initiative No. GDZX2303007, HKU Seed Fund for Basic Research for New Staff via Project 2201100596. Z.S. also acknowledges the support of HK Institute of Quantum Science and Technology. N.G. acknowledges support through the Research Excellence Scholarship from SandboxAQ. T.L. acknowledges the support from the National Natural Science Foundation of China (NSFC) via Project No. 62372006 and No. 92365117.
\end{acknowledgments}

\bibliography{ref}

\begin{thebibliography}{58}%
\makeatletter
\providecommand \@ifxundefined [1]{%
 \@ifx{#1\undefined}
}%
\providecommand \@ifnum [1]{%
 \ifnum #1\expandafter \@firstoftwo
 \else \expandafter \@secondoftwo
 \fi
}%
\providecommand \@ifx [1]{%
 \ifx #1\expandafter \@firstoftwo
 \else \expandafter \@secondoftwo
 \fi
}%
\providecommand \natexlab [1]{#1}%
\providecommand \enquote  [1]{``#1''}%
\providecommand \bibnamefont  [1]{#1}%
\providecommand \bibfnamefont [1]{#1}%
\providecommand \citenamefont [1]{#1}%
\providecommand \href@noop [0]{\@secondoftwo}%
\providecommand \href [0]{\begingroup \@sanitize@url \@href}%
\providecommand \@href[1]{\@@startlink{#1}\@@href}%
\providecommand \@@href[1]{\endgroup#1\@@endlink}%
\providecommand \@sanitize@url [0]{\catcode `\\12\catcode `\$12\catcode `\&12\catcode `\#12\catcode `\^12\catcode `\_12\catcode `\%12\relax}%
\providecommand \@@startlink[1]{}%
\providecommand \@@endlink[0]{}%
\providecommand \url  [0]{\begingroup\@sanitize@url \@url }%
\providecommand \@url [1]{\endgroup\@href {#1}{\urlprefix }}%
\providecommand \urlprefix  [0]{URL }%
\providecommand \Eprint [0]{\href }%
\providecommand \doibase [0]{https://doi.org/}%
\providecommand \selectlanguage [0]{\@gobble}%
\providecommand \bibinfo  [0]{\@secondoftwo}%
\providecommand \bibfield  [0]{\@secondoftwo}%
\providecommand \translation [1]{[#1]}%
\providecommand \BibitemOpen [0]{}%
\providecommand \bibitemStop [0]{}%
\providecommand \bibitemNoStop [0]{.\EOS\space}%
\providecommand \EOS [0]{\spacefactor3000\relax}%
\providecommand \BibitemShut  [1]{\csname bibitem#1\endcsname}%
\let\auto@bib@innerbib\@empty
\bibitem [{\citenamefont {Nielsen}\ and\ \citenamefont {Chuang}(2010)}]{nielsen2010quantum}%
  \BibitemOpen
  \bibfield  {author} {\bibinfo {author} {\bibfnamefont {M.~A.}\ \bibnamefont {Nielsen}}\ and\ \bibinfo {author} {\bibfnamefont {I.~L.}\ \bibnamefont {Chuang}},\ }\href@noop {} {\emph {\bibinfo {title} {Quantum computation and quantum information}}}\ (\bibinfo  {publisher} {Cambridge university press},\ \bibinfo {year} {2010})\BibitemShut {NoStop}%
\bibitem [{\citenamefont {Kitaev}(1995)}]{kitaev1995quantum}%
  \BibitemOpen
  \bibfield  {author} {\bibinfo {author} {\bibfnamefont {A.~Y.}\ \bibnamefont {Kitaev}},\ }\bibfield  {title} {\bibinfo {title} {Quantum measurements and the abelian stabilizer problem},\ }\href@noop {} {\bibfield  {journal} {\bibinfo  {journal} {arXiv preprint quant-ph/9511026}\ } (\bibinfo {year} {1995})}\BibitemShut {NoStop}%
\bibitem [{\citenamefont {Dalzell}\ \emph {et~al.}(2023)\citenamefont {Dalzell}, \citenamefont {McArdle}, \citenamefont {Berta}, \citenamefont {Bienias}, \citenamefont {Chen}, \citenamefont {Gily{\'e}n}, \citenamefont {Hann}, \citenamefont {Kastoryano}, \citenamefont {Khabiboulline}, \citenamefont {Kubica}, \citenamefont {Salton}, \citenamefont {Wang},\ and\ \citenamefont {Brandão}}]{dalzell2023quantum}%
  \BibitemOpen
  \bibfield  {author} {\bibinfo {author} {\bibfnamefont {A.~M.}\ \bibnamefont {Dalzell}}, \bibinfo {author} {\bibfnamefont {S.}~\bibnamefont {McArdle}}, \bibinfo {author} {\bibfnamefont {M.}~\bibnamefont {Berta}}, \bibinfo {author} {\bibfnamefont {P.}~\bibnamefont {Bienias}}, \bibinfo {author} {\bibfnamefont {C.-F.}\ \bibnamefont {Chen}}, \bibinfo {author} {\bibfnamefont {A.}~\bibnamefont {Gily{\'e}n}}, \bibinfo {author} {\bibfnamefont {C.~T.}\ \bibnamefont {Hann}}, \bibinfo {author} {\bibfnamefont {M.~J.}\ \bibnamefont {Kastoryano}}, \bibinfo {author} {\bibfnamefont {E.~T.}\ \bibnamefont {Khabiboulline}}, \bibinfo {author} {\bibfnamefont {A.}~\bibnamefont {Kubica}}, \bibinfo {author} {\bibfnamefont {G.}~\bibnamefont {Salton}}, \bibinfo {author} {\bibfnamefont {S.}~\bibnamefont {Wang}},\ and\ \bibinfo {author} {\bibfnamefont {F.~G. S.~L.}\ \bibnamefont {Brandão}},\ }\bibfield  {title} {\bibinfo {title} {Quantum algorithms: A survey of applications and end-to-end complexities},\ }\href@noop {} {\bibfield
  {journal} {\bibinfo  {journal} {arXiv preprint arXiv:2310.03011}\ } (\bibinfo {year} {2023})}\BibitemShut {NoStop}%
\bibitem [{\citenamefont {Wiebe}\ and\ \citenamefont {Granade}(2016)}]{wiebe2016efficient}%
  \BibitemOpen
  \bibfield  {author} {\bibinfo {author} {\bibfnamefont {N.}~\bibnamefont {Wiebe}}\ and\ \bibinfo {author} {\bibfnamefont {C.}~\bibnamefont {Granade}},\ }\bibfield  {title} {\bibinfo {title} {Efficient {B}ayesian phase estimation},\ }\href@noop {} {\bibfield  {journal} {\bibinfo  {journal} {Physical Review Letters}\ }\textbf {\bibinfo {volume} {117}},\ \bibinfo {pages} {010503} (\bibinfo {year} {2016})}\BibitemShut {NoStop}%
\bibitem [{\citenamefont {Ding}\ and\ \citenamefont {Lin}(2023)}]{ding2023even}%
  \BibitemOpen
  \bibfield  {author} {\bibinfo {author} {\bibfnamefont {Z.}~\bibnamefont {Ding}}\ and\ \bibinfo {author} {\bibfnamefont {L.}~\bibnamefont {Lin}},\ }\bibfield  {title} {\bibinfo {title} {Even shorter quantum circuit for phase estimation on early fault-tolerant quantum computers with applications to ground-state energy estimation},\ }\href@noop {} {\bibfield  {journal} {\bibinfo  {journal} {PRX Quantum}\ }\textbf {\bibinfo {volume} {4}},\ \bibinfo {pages} {020331} (\bibinfo {year} {2023})}\BibitemShut {NoStop}%
\bibitem [{\citenamefont {Higgins}\ \emph {et~al.}(2007)\citenamefont {Higgins}, \citenamefont {Berry}, \citenamefont {Bartlett}, \citenamefont {Wiseman},\ and\ \citenamefont {Pryde}}]{higgins2007entanglement}%
  \BibitemOpen
  \bibfield  {author} {\bibinfo {author} {\bibfnamefont {B.~L.}\ \bibnamefont {Higgins}}, \bibinfo {author} {\bibfnamefont {D.~W.}\ \bibnamefont {Berry}}, \bibinfo {author} {\bibfnamefont {S.~D.}\ \bibnamefont {Bartlett}}, \bibinfo {author} {\bibfnamefont {H.~M.}\ \bibnamefont {Wiseman}},\ and\ \bibinfo {author} {\bibfnamefont {G.~J.}\ \bibnamefont {Pryde}},\ }\bibfield  {title} {\bibinfo {title} {Entanglement-free {H}eisenberg-limited phase estimation},\ }\href@noop {} {\bibfield  {journal} {\bibinfo  {journal} {Nature}\ }\textbf {\bibinfo {volume} {450}},\ \bibinfo {pages} {393} (\bibinfo {year} {2007})}\BibitemShut {NoStop}%
\bibitem [{\citenamefont {Shor}(1994)}]{shor1994algorithms}%
  \BibitemOpen
  \bibfield  {author} {\bibinfo {author} {\bibfnamefont {P.~W.}\ \bibnamefont {Shor}},\ }\bibfield  {title} {\bibinfo {title} {Algorithms for quantum computation: discrete logarithms and factoring},\ }in\ \href@noop {} {\emph {\bibinfo {booktitle} {Proceedings 35th annual symposium on foundations of computer science}}}\ (\bibinfo {organization} {Ieee},\ \bibinfo {year} {1994})\ pp.\ \bibinfo {pages} {124--134}\BibitemShut {NoStop}%
\bibitem [{\citenamefont {Harrow}\ \emph {et~al.}(2009)\citenamefont {Harrow}, \citenamefont {Hassidim},\ and\ \citenamefont {Lloyd}}]{harrow2009quantum}%
  \BibitemOpen
  \bibfield  {author} {\bibinfo {author} {\bibfnamefont {A.~W.}\ \bibnamefont {Harrow}}, \bibinfo {author} {\bibfnamefont {A.}~\bibnamefont {Hassidim}},\ and\ \bibinfo {author} {\bibfnamefont {S.}~\bibnamefont {Lloyd}},\ }\bibfield  {title} {\bibinfo {title} {Quantum algorithm for linear systems of equations},\ }\href@noop {} {\bibfield  {journal} {\bibinfo  {journal} {Physical Review Letters}\ }\textbf {\bibinfo {volume} {103}},\ \bibinfo {pages} {150502} (\bibinfo {year} {2009})}\BibitemShut {NoStop}%
\bibitem [{\citenamefont {McArdle}\ \emph {et~al.}(2020)\citenamefont {McArdle}, \citenamefont {Endo}, \citenamefont {Aspuru-Guzik}, \citenamefont {Benjamin},\ and\ \citenamefont {Yuan}}]{mcardle2020quantum}%
  \BibitemOpen
  \bibfield  {author} {\bibinfo {author} {\bibfnamefont {S.}~\bibnamefont {McArdle}}, \bibinfo {author} {\bibfnamefont {S.}~\bibnamefont {Endo}}, \bibinfo {author} {\bibfnamefont {A.}~\bibnamefont {Aspuru-Guzik}}, \bibinfo {author} {\bibfnamefont {S.~C.}\ \bibnamefont {Benjamin}},\ and\ \bibinfo {author} {\bibfnamefont {X.}~\bibnamefont {Yuan}},\ }\bibfield  {title} {\bibinfo {title} {Quantum computational chemistry},\ }\href@noop {} {\bibfield  {journal} {\bibinfo  {journal} {Reviews of Modern Physics}\ }\textbf {\bibinfo {volume} {92}},\ \bibinfo {pages} {015003} (\bibinfo {year} {2020})}\BibitemShut {NoStop}%
\bibitem [{\citenamefont {Feynman}(2018)}]{feynman2018simulating}%
  \BibitemOpen
  \bibfield  {author} {\bibinfo {author} {\bibfnamefont {R.~P.}\ \bibnamefont {Feynman}},\ }\bibfield  {title} {\bibinfo {title} {Simulating physics with computers},\ }in\ \href@noop {} {\emph {\bibinfo {booktitle} {Feynman and computation}}}\ (\bibinfo  {publisher} {cRc Press},\ \bibinfo {year} {2018})\ pp.\ \bibinfo {pages} {133--153}\BibitemShut {NoStop}%
\bibitem [{\citenamefont {Zwierz}\ \emph {et~al.}(2010)\citenamefont {Zwierz}, \citenamefont {P{\'e}rez-Delgado},\ and\ \citenamefont {Kok}}]{zwierz2010general}%
  \BibitemOpen
  \bibfield  {author} {\bibinfo {author} {\bibfnamefont {M.}~\bibnamefont {Zwierz}}, \bibinfo {author} {\bibfnamefont {C.~A.}\ \bibnamefont {P{\'e}rez-Delgado}},\ and\ \bibinfo {author} {\bibfnamefont {P.}~\bibnamefont {Kok}},\ }\bibfield  {title} {\bibinfo {title} {General optimality of the heisenberg limit for quantum metrology},\ }\href@noop {} {\bibfield  {journal} {\bibinfo  {journal} {Physical Review Letters}\ }\textbf {\bibinfo {volume} {105}},\ \bibinfo {pages} {180402} (\bibinfo {year} {2010})}\BibitemShut {NoStop}%
\bibitem [{\citenamefont {Busch}\ \emph {et~al.}(2007)\citenamefont {Busch}, \citenamefont {Heinonen},\ and\ \citenamefont {Lahti}}]{busch2007heisenberg}%
  \BibitemOpen
  \bibfield  {author} {\bibinfo {author} {\bibfnamefont {P.}~\bibnamefont {Busch}}, \bibinfo {author} {\bibfnamefont {T.}~\bibnamefont {Heinonen}},\ and\ \bibinfo {author} {\bibfnamefont {P.}~\bibnamefont {Lahti}},\ }\bibfield  {title} {\bibinfo {title} {Heisenberg's uncertainty principle},\ }\href@noop {} {\bibfield  {journal} {\bibinfo  {journal} {Physics Reports}\ }\textbf {\bibinfo {volume} {452}},\ \bibinfo {pages} {155} (\bibinfo {year} {2007})}\BibitemShut {NoStop}%
\bibitem [{\citenamefont {Kijowski}(1974)}]{kijowski1974time}%
  \BibitemOpen
  \bibfield  {author} {\bibinfo {author} {\bibfnamefont {J.}~\bibnamefont {Kijowski}},\ }\bibfield  {title} {\bibinfo {title} {On the time operator in quantum mechanics and the heisenberg uncertainty relation for energy and time},\ }\href@noop {} {\bibfield  {journal} {\bibinfo  {journal} {Reports on Mathematical Physics}\ }\textbf {\bibinfo {volume} {6}},\ \bibinfo {pages} {361} (\bibinfo {year} {1974})}\BibitemShut {NoStop}%
\bibitem [{\citenamefont {Atia}\ and\ \citenamefont {Aharonov}(2017)}]{atia2017fast}%
  \BibitemOpen
  \bibfield  {author} {\bibinfo {author} {\bibfnamefont {Y.}~\bibnamefont {Atia}}\ and\ \bibinfo {author} {\bibfnamefont {D.}~\bibnamefont {Aharonov}},\ }\bibfield  {title} {\bibinfo {title} {Fast-forwarding of {H}amiltonians and exponentially precise measurements},\ }\href@noop {} {\bibfield  {journal} {\bibinfo  {journal} {Nature Communications}\ }\textbf {\bibinfo {volume} {8}},\ \bibinfo {pages} {1572} (\bibinfo {year} {2017})}\BibitemShut {NoStop}%
\bibitem [{\citenamefont {Berry}\ \emph {et~al.}(2007)\citenamefont {Berry}, \citenamefont {Ahokas}, \citenamefont {Cleve},\ and\ \citenamefont {Sanders}}]{berry2007efficient}%
  \BibitemOpen
  \bibfield  {author} {\bibinfo {author} {\bibfnamefont {D.~W.}\ \bibnamefont {Berry}}, \bibinfo {author} {\bibfnamefont {G.}~\bibnamefont {Ahokas}}, \bibinfo {author} {\bibfnamefont {R.}~\bibnamefont {Cleve}},\ and\ \bibinfo {author} {\bibfnamefont {B.~C.}\ \bibnamefont {Sanders}},\ }\bibfield  {title} {\bibinfo {title} {Efficient quantum algorithms for simulating sparse {H}amiltonians},\ }\href@noop {} {\bibfield  {journal} {\bibinfo  {journal} {Communications in Mathematical Physics}\ }\textbf {\bibinfo {volume} {270}},\ \bibinfo {pages} {359} (\bibinfo {year} {2007})}\BibitemShut {NoStop}%
\bibitem [{\citenamefont {Lloyd}(1996)}]{lloyd1996universal}%
  \BibitemOpen
  \bibfield  {author} {\bibinfo {author} {\bibfnamefont {S.}~\bibnamefont {Lloyd}},\ }\bibfield  {title} {\bibinfo {title} {Universal quantum simulators},\ }\href@noop {} {\bibfield  {journal} {\bibinfo  {journal} {Science}\ }\textbf {\bibinfo {volume} {273}},\ \bibinfo {pages} {1073} (\bibinfo {year} {1996})}\BibitemShut {NoStop}%
\bibitem [{\citenamefont {Berry}\ \emph {et~al.}(2015)\citenamefont {Berry}, \citenamefont {Childs}, \citenamefont {Cleve}, \citenamefont {Kothari},\ and\ \citenamefont {Somma}}]{berry2015simulating}%
  \BibitemOpen
  \bibfield  {author} {\bibinfo {author} {\bibfnamefont {D.~W.}\ \bibnamefont {Berry}}, \bibinfo {author} {\bibfnamefont {A.~M.}\ \bibnamefont {Childs}}, \bibinfo {author} {\bibfnamefont {R.}~\bibnamefont {Cleve}}, \bibinfo {author} {\bibfnamefont {R.}~\bibnamefont {Kothari}},\ and\ \bibinfo {author} {\bibfnamefont {R.~D.}\ \bibnamefont {Somma}},\ }\bibfield  {title} {\bibinfo {title} {Simulating hamiltonian dynamics with a truncated taylor series},\ }\href@noop {} {\bibfield  {journal} {\bibinfo  {journal} {Physical review letters}\ }\textbf {\bibinfo {volume} {114}},\ \bibinfo {pages} {090502} (\bibinfo {year} {2015})}\BibitemShut {NoStop}%
\bibitem [{\citenamefont {Low}\ and\ \citenamefont {Chuang}(2017)}]{low2017optimal}%
  \BibitemOpen
  \bibfield  {author} {\bibinfo {author} {\bibfnamefont {G.~H.}\ \bibnamefont {Low}}\ and\ \bibinfo {author} {\bibfnamefont {I.~L.}\ \bibnamefont {Chuang}},\ }\bibfield  {title} {\bibinfo {title} {Optimal {H}amiltonian simulation by quantum signal processing},\ }\href@noop {} {\bibfield  {journal} {\bibinfo  {journal} {Physical Review Letters}\ }\textbf {\bibinfo {volume} {118}},\ \bibinfo {pages} {010501} (\bibinfo {year} {2017})}\BibitemShut {NoStop}%
\bibitem [{\citenamefont {Lindblad}(1976)}]{lindblad1976generators}%
  \BibitemOpen
  \bibfield  {author} {\bibinfo {author} {\bibfnamefont {G.}~\bibnamefont {Lindblad}},\ }\bibfield  {title} {\bibinfo {title} {On the generators of quantum dynamical semigroups},\ }\href@noop {} {\bibfield  {journal} {\bibinfo  {journal} {Communications in Mathematical Physics}\ }\textbf {\bibinfo {volume} {48}},\ \bibinfo {pages} {119} (\bibinfo {year} {1976})}\BibitemShut {NoStop}%
\bibitem [{\citenamefont {Chen}\ \emph {et~al.}(2023{\natexlab{a}})\citenamefont {Chen}, \citenamefont {Kastoryano}, \citenamefont {Brand{\~a}o},\ and\ \citenamefont {Gily{\'e}n}}]{chen2023quantum}%
  \BibitemOpen
  \bibfield  {author} {\bibinfo {author} {\bibfnamefont {C.-F.}\ \bibnamefont {Chen}}, \bibinfo {author} {\bibfnamefont {M.~J.}\ \bibnamefont {Kastoryano}}, \bibinfo {author} {\bibfnamefont {F.~G. S.~L.}\ \bibnamefont {Brand{\~a}o}},\ and\ \bibinfo {author} {\bibfnamefont {A.}~\bibnamefont {Gily{\'e}n}},\ }\bibfield  {title} {\bibinfo {title} {Quantum thermal state preparation},\ }\href@noop {} {\bibfield  {journal} {\bibinfo  {journal} {arXiv preprint arXiv:2303.18224}\ } (\bibinfo {year} {2023}{\natexlab{a}})}\BibitemShut {NoStop}%
\bibitem [{\citenamefont {Chen}\ \emph {et~al.}(2023{\natexlab{b}})\citenamefont {Chen}, \citenamefont {Kastoryano},\ and\ \citenamefont {Gily{\'e}n}}]{chen2023efficient}%
  \BibitemOpen
  \bibfield  {author} {\bibinfo {author} {\bibfnamefont {C.-F.}\ \bibnamefont {Chen}}, \bibinfo {author} {\bibfnamefont {M.~J.}\ \bibnamefont {Kastoryano}},\ and\ \bibinfo {author} {\bibfnamefont {A.}~\bibnamefont {Gily{\'e}n}},\ }\bibfield  {title} {\bibinfo {title} {An efficient and exact noncommutative quantum {G}ibbs sampler},\ }\href@noop {} {\bibfield  {journal} {\bibinfo  {journal} {arXiv preprint arXiv:2311.09207}\ } (\bibinfo {year} {2023}{\natexlab{b}})}\BibitemShut {NoStop}%
\bibitem [{\citenamefont {Rouz{\'e}}\ \emph {et~al.}(2024)\citenamefont {Rouz{\'e}}, \citenamefont {Fran{\c{c}}a},\ and\ \citenamefont {Alhambra}}]{rouze2024optimal}%
  \BibitemOpen
  \bibfield  {author} {\bibinfo {author} {\bibfnamefont {C.}~\bibnamefont {Rouz{\'e}}}, \bibinfo {author} {\bibfnamefont {D.~S.}\ \bibnamefont {Fran{\c{c}}a}},\ and\ \bibinfo {author} {\bibfnamefont {{\'A}.~M.}\ \bibnamefont {Alhambra}},\ }\bibfield  {title} {\bibinfo {title} {Optimal quantum algorithm for {G}ibbs state preparation},\ }\href@noop {} {\bibfield  {journal} {\bibinfo  {journal} {arXiv preprint arXiv:2411.04885}\ } (\bibinfo {year} {2024})}\BibitemShut {NoStop}%
\bibitem [{\citenamefont {Ding}\ \emph {et~al.}(2024)\citenamefont {Ding}, \citenamefont {Chen},\ and\ \citenamefont {Lin}}]{ding2024single}%
  \BibitemOpen
  \bibfield  {author} {\bibinfo {author} {\bibfnamefont {Z.}~\bibnamefont {Ding}}, \bibinfo {author} {\bibfnamefont {C.-F.}\ \bibnamefont {Chen}},\ and\ \bibinfo {author} {\bibfnamefont {L.}~\bibnamefont {Lin}},\ }\bibfield  {title} {\bibinfo {title} {Single-ancilla ground state preparation via lindbladians},\ }\href@noop {} {\bibfield  {journal} {\bibinfo  {journal} {Physical Review Research}\ }\textbf {\bibinfo {volume} {6}},\ \bibinfo {pages} {033147} (\bibinfo {year} {2024})}\BibitemShut {NoStop}%
\bibitem [{\citenamefont {Chen}\ \emph {et~al.}(2025)\citenamefont {Chen}, \citenamefont {Lu}, \citenamefont {Wang}, \citenamefont {Liu},\ and\ \citenamefont {Li}}]{chen2025quantum}%
  \BibitemOpen
  \bibfield  {author} {\bibinfo {author} {\bibfnamefont {Z.}~\bibnamefont {Chen}}, \bibinfo {author} {\bibfnamefont {Y.}~\bibnamefont {Lu}}, \bibinfo {author} {\bibfnamefont {H.}~\bibnamefont {Wang}}, \bibinfo {author} {\bibfnamefont {Y.}~\bibnamefont {Liu}},\ and\ \bibinfo {author} {\bibfnamefont {T.}~\bibnamefont {Li}},\ }\bibfield  {title} {\bibinfo {title} {Quantum langevin dynamics for optimization},\ }\href@noop {} {\bibfield  {journal} {\bibinfo  {journal} {Communications in Mathematical Physics}\ }\textbf {\bibinfo {volume} {406}},\ \bibinfo {pages} {52} (\bibinfo {year} {2025})}\BibitemShut {NoStop}%
\bibitem [{\citenamefont {He}\ \emph {et~al.}(2024)\citenamefont {He}, \citenamefont {Li}, \citenamefont {Li}, \citenamefont {Li}, \citenamefont {Wang},\ and\ \citenamefont {Wang}}]{he2024efficient}%
  \BibitemOpen
  \bibfield  {author} {\bibinfo {author} {\bibfnamefont {W.}~\bibnamefont {He}}, \bibinfo {author} {\bibfnamefont {T.}~\bibnamefont {Li}}, \bibinfo {author} {\bibfnamefont {X.}~\bibnamefont {Li}}, \bibinfo {author} {\bibfnamefont {Z.}~\bibnamefont {Li}}, \bibinfo {author} {\bibfnamefont {C.}~\bibnamefont {Wang}},\ and\ \bibinfo {author} {\bibfnamefont {K.}~\bibnamefont {Wang}},\ }\bibfield  {title} {\bibinfo {title} {Efficient optimal control of open quantum systems},\ }in\ \href@noop {} {\emph {\bibinfo {booktitle} {19th Conference on the Theory of Quantum Computation, Communication and Cryptography (TQC 2024)}}}\ (\bibinfo {organization} {Schloss Dagstuhl--Leibniz-Zentrum f{\"u}r Informatik},\ \bibinfo {year} {2024})\ pp.\ \bibinfo {pages} {3--1}\BibitemShut {NoStop}%
\bibitem [{\citenamefont {Shang}\ \emph {et~al.}(2024)\citenamefont {Shang}, \citenamefont {Guo}, \citenamefont {An},\ and\ \citenamefont {Zhao}}]{shang2024design}%
  \BibitemOpen
  \bibfield  {author} {\bibinfo {author} {\bibfnamefont {Z.-X.}\ \bibnamefont {Shang}}, \bibinfo {author} {\bibfnamefont {N.}~\bibnamefont {Guo}}, \bibinfo {author} {\bibfnamefont {D.}~\bibnamefont {An}},\ and\ \bibinfo {author} {\bibfnamefont {Q.}~\bibnamefont {Zhao}},\ }\bibfield  {title} {\bibinfo {title} {Design nearly optimal quantum algorithm for linear differential equations via lindbladians},\ }\href@noop {} {\bibfield  {journal} {\bibinfo  {journal} {arXiv preprint arXiv:2410.19628}\ } (\bibinfo {year} {2024})}\BibitemShut {NoStop}%
\bibitem [{\citenamefont {Kliesch}\ \emph {et~al.}(2011)\citenamefont {Kliesch}, \citenamefont {Barthel}, \citenamefont {Gogolin}, \citenamefont {Kastoryano},\ and\ \citenamefont {Eisert}}]{kliesch2011dissipative}%
  \BibitemOpen
  \bibfield  {author} {\bibinfo {author} {\bibfnamefont {M.}~\bibnamefont {Kliesch}}, \bibinfo {author} {\bibfnamefont {T.}~\bibnamefont {Barthel}}, \bibinfo {author} {\bibfnamefont {C.}~\bibnamefont {Gogolin}}, \bibinfo {author} {\bibfnamefont {M.}~\bibnamefont {Kastoryano}},\ and\ \bibinfo {author} {\bibfnamefont {J.}~\bibnamefont {Eisert}},\ }\bibfield  {title} {\bibinfo {title} {Dissipative quantum {Church-Turing} theorem},\ }\href@noop {} {\bibfield  {journal} {\bibinfo  {journal} {Physical Review Letters}\ }\textbf {\bibinfo {volume} {107}},\ \bibinfo {pages} {120501} (\bibinfo {year} {2011})}\BibitemShut {NoStop}%
\bibitem [{\citenamefont {Cleve}\ and\ \citenamefont {Wang}(2017)}]{cleve2017efficient}%
  \BibitemOpen
  \bibfield  {author} {\bibinfo {author} {\bibfnamefont {R.}~\bibnamefont {Cleve}}\ and\ \bibinfo {author} {\bibfnamefont {C.}~\bibnamefont {Wang}},\ }\bibfield  {title} {\bibinfo {title} {Efficient quantum algorithms for simulating {L}indblad evolution},\ }in\ \href@noop {} {\emph {\bibinfo {booktitle} {44th International Colloquium on Automata, Languages, and Programming (ICALP 2017)}}}\ (\bibinfo {organization} {Schloss Dagstuhl--Leibniz-Zentrum f{\"u}r Informatik},\ \bibinfo {year} {2017})\ pp.\ \bibinfo {pages} {17--1}\BibitemShut {NoStop}%
\bibitem [{\citenamefont {Li}\ and\ \citenamefont {Wang}(2022)}]{li2022simulating}%
  \BibitemOpen
  \bibfield  {author} {\bibinfo {author} {\bibfnamefont {X.}~\bibnamefont {Li}}\ and\ \bibinfo {author} {\bibfnamefont {C.}~\bibnamefont {Wang}},\ }\bibfield  {title} {\bibinfo {title} {Simulating markovian open quantum systems using higher-order series expansion},\ }\href@noop {} {\bibfield  {journal} {\bibinfo  {journal} {arXiv preprint arXiv:2212.02051}\ } (\bibinfo {year} {2022})}\BibitemShut {NoStop}%
\bibitem [{\citenamefont {Ding}\ \emph {et~al.}(2025)\citenamefont {Ding}, \citenamefont {Junge}, \citenamefont {Schleich},\ and\ \citenamefont {Wu}}]{ding2025lower}%
  \BibitemOpen
  \bibfield  {author} {\bibinfo {author} {\bibfnamefont {Z.}~\bibnamefont {Ding}}, \bibinfo {author} {\bibfnamefont {M.}~\bibnamefont {Junge}}, \bibinfo {author} {\bibfnamefont {P.}~\bibnamefont {Schleich}},\ and\ \bibinfo {author} {\bibfnamefont {P.}~\bibnamefont {Wu}},\ }\bibfield  {title} {\bibinfo {title} {Lower bound for simulation cost of open quantum systems: Lipschitz continuity approach},\ }\href@noop {} {\bibfield  {journal} {\bibinfo  {journal} {Communications in Mathematical Physics}\ }\textbf {\bibinfo {volume} {406}},\ \bibinfo {pages} {60} (\bibinfo {year} {2025})}\BibitemShut {NoStop}%
\bibitem [{\citenamefont {Gu}\ \emph {et~al.}(2021)\citenamefont {Gu}, \citenamefont {Somma},\ and\ \citenamefont {{\c{S}}ahino{\u{g}}lu}}]{gu2021fast}%
  \BibitemOpen
  \bibfield  {author} {\bibinfo {author} {\bibfnamefont {S.}~\bibnamefont {Gu}}, \bibinfo {author} {\bibfnamefont {R.~D.}\ \bibnamefont {Somma}},\ and\ \bibinfo {author} {\bibfnamefont {B.}~\bibnamefont {{\c{S}}ahino{\u{g}}lu}},\ }\bibfield  {title} {\bibinfo {title} {Fast-forwarding quantum evolution},\ }\href@noop {} {\bibfield  {journal} {\bibinfo  {journal} {Quantum}\ }\textbf {\bibinfo {volume} {5}},\ \bibinfo {pages} {577} (\bibinfo {year} {2021})}\BibitemShut {NoStop}%
\bibitem [{\citenamefont {Xia}\ \emph {et~al.}(2019)\citenamefont {Xia}, \citenamefont {Liu}, \citenamefont {Nie}, \citenamefont {Fu}, \citenamefont {Wan},\ and\ \citenamefont {Kong}}]{xia2019random}%
  \BibitemOpen
  \bibfield  {author} {\bibinfo {author} {\bibfnamefont {F.}~\bibnamefont {Xia}}, \bibinfo {author} {\bibfnamefont {J.}~\bibnamefont {Liu}}, \bibinfo {author} {\bibfnamefont {H.}~\bibnamefont {Nie}}, \bibinfo {author} {\bibfnamefont {Y.}~\bibnamefont {Fu}}, \bibinfo {author} {\bibfnamefont {L.}~\bibnamefont {Wan}},\ and\ \bibinfo {author} {\bibfnamefont {X.}~\bibnamefont {Kong}},\ }\bibfield  {title} {\bibinfo {title} {Random walks: A review of algorithms and applications},\ }\href@noop {} {\bibfield  {journal} {\bibinfo  {journal} {IEEE Transactions on Emerging Topics in Computational Intelligence}\ }\textbf {\bibinfo {volume} {4}},\ \bibinfo {pages} {95} (\bibinfo {year} {2019})}\BibitemShut {NoStop}%
\bibitem [{\citenamefont {Grover}(1996)}]{grover1996fast}%
  \BibitemOpen
  \bibfield  {author} {\bibinfo {author} {\bibfnamefont {L.~K.}\ \bibnamefont {Grover}},\ }\bibfield  {title} {\bibinfo {title} {A fast quantum mechanical algorithm for database search},\ }in\ \href@noop {} {\emph {\bibinfo {booktitle} {Proceedings of the 28th Annual ACM Symposium on Theory of Computing}}}\ (\bibinfo {year} {1996})\ pp.\ \bibinfo {pages} {212--219}\BibitemShut {NoStop}%
\bibitem [{\citenamefont {Brassard}\ \emph {et~al.}(2000)\citenamefont {Brassard}, \citenamefont {Hoyer}, \citenamefont {Mosca},\ and\ \citenamefont {Tapp}}]{brassard2000quantum}%
  \BibitemOpen
  \bibfield  {author} {\bibinfo {author} {\bibfnamefont {G.}~\bibnamefont {Brassard}}, \bibinfo {author} {\bibfnamefont {P.}~\bibnamefont {Hoyer}}, \bibinfo {author} {\bibfnamefont {M.}~\bibnamefont {Mosca}},\ and\ \bibinfo {author} {\bibfnamefont {A.}~\bibnamefont {Tapp}},\ }\bibfield  {title} {\bibinfo {title} {Quantum amplitude amplification and estimation},\ }\href@noop {} {\bibfield  {journal} {\bibinfo  {journal} {arXiv preprint quant-ph/0005055}\ } (\bibinfo {year} {2000})}\BibitemShut {NoStop}%
\bibitem [{\citenamefont {Gily{\'e}n}\ \emph {et~al.}(2019)\citenamefont {Gily{\'e}n}, \citenamefont {Su}, \citenamefont {Low},\ and\ \citenamefont {Wiebe}}]{gilyen2019quantum}%
  \BibitemOpen
  \bibfield  {author} {\bibinfo {author} {\bibfnamefont {A.}~\bibnamefont {Gily{\'e}n}}, \bibinfo {author} {\bibfnamefont {Y.}~\bibnamefont {Su}}, \bibinfo {author} {\bibfnamefont {G.~H.}\ \bibnamefont {Low}},\ and\ \bibinfo {author} {\bibfnamefont {N.}~\bibnamefont {Wiebe}},\ }\bibfield  {title} {\bibinfo {title} {Quantum singular value transformation and beyond: exponential improvements for quantum matrix arithmetics},\ }in\ \href@noop {} {\emph {\bibinfo {booktitle} {Proceedings of the 51st Annual ACM SIGACT Symposium on Theory of Computing}}}\ (\bibinfo {year} {2019})\ pp.\ \bibinfo {pages} {193--204}\BibitemShut {NoStop}%
\bibitem [{\citenamefont {Shang}\ \emph {et~al.}(2025)\citenamefont {Shang}, \citenamefont {An},\ and\ \citenamefont {Shao}}]{shang2025exponential}%
  \BibitemOpen
  \bibfield  {author} {\bibinfo {author} {\bibfnamefont {Z.-X.}\ \bibnamefont {Shang}}, \bibinfo {author} {\bibfnamefont {D.}~\bibnamefont {An}},\ and\ \bibinfo {author} {\bibfnamefont {C.}~\bibnamefont {Shao}},\ }\bibfield  {title} {\bibinfo {title} {Exponential lindbladian fast forwarding and exponential amplification of certain gibbs state properties},\ }\href@noop {} {\bibfield  {journal} {\bibinfo  {journal} {arXiv preprint arXiv:2509.09517}\ } (\bibinfo {year} {2025})}\BibitemShut {NoStop}%
\bibitem [{\citenamefont {Hozo}\ \emph {et~al.}(2005)\citenamefont {Hozo}, \citenamefont {Djulbegovic},\ and\ \citenamefont {Hozo}}]{hozo2005estimating}%
  \BibitemOpen
  \bibfield  {author} {\bibinfo {author} {\bibfnamefont {S.~P.}\ \bibnamefont {Hozo}}, \bibinfo {author} {\bibfnamefont {B.}~\bibnamefont {Djulbegovic}},\ and\ \bibinfo {author} {\bibfnamefont {I.}~\bibnamefont {Hozo}},\ }\bibfield  {title} {\bibinfo {title} {Estimating the mean and variance from the median, range, and the size of a sample},\ }\href@noop {} {\bibfield  {journal} {\bibinfo  {journal} {BMC medical research methodology}\ }\textbf {\bibinfo {volume} {5}},\ \bibinfo {pages} {13} (\bibinfo {year} {2005})}\BibitemShut {NoStop}%
\bibitem [{\citenamefont {Childs}(2010)}]{childs2010relationship}%
  \BibitemOpen
  \bibfield  {author} {\bibinfo {author} {\bibfnamefont {A.~M.}\ \bibnamefont {Childs}},\ }\bibfield  {title} {\bibinfo {title} {On the relationship between continuous-and discrete-time quantum walk},\ }\href@noop {} {\bibfield  {journal} {\bibinfo  {journal} {Communications in Mathematical Physics}\ }\textbf {\bibinfo {volume} {294}},\ \bibinfo {pages} {581} (\bibinfo {year} {2010})}\BibitemShut {NoStop}%
\bibitem [{\citenamefont {Lin}\ and\ \citenamefont {Tong}(2020)}]{lin2020near}%
  \BibitemOpen
  \bibfield  {author} {\bibinfo {author} {\bibfnamefont {L.}~\bibnamefont {Lin}}\ and\ \bibinfo {author} {\bibfnamefont {Y.}~\bibnamefont {Tong}},\ }\bibfield  {title} {\bibinfo {title} {Near-optimal ground state preparation},\ }\href@noop {} {\bibfield  {journal} {\bibinfo  {journal} {Quantum}\ }\textbf {\bibinfo {volume} {4}},\ \bibinfo {pages} {372} (\bibinfo {year} {2020})}\BibitemShut {NoStop}%
\bibitem [{\citenamefont {Lin}\ and\ \citenamefont {Tong}(2022)}]{lin2022heisenberg}%
  \BibitemOpen
  \bibfield  {author} {\bibinfo {author} {\bibfnamefont {L.}~\bibnamefont {Lin}}\ and\ \bibinfo {author} {\bibfnamefont {Y.}~\bibnamefont {Tong}},\ }\bibfield  {title} {\bibinfo {title} {Heisenberg-limited ground-state energy estimation for early fault-tolerant quantum computers},\ }\href@noop {} {\bibfield  {journal} {\bibinfo  {journal} {PRX quantum}\ }\textbf {\bibinfo {volume} {3}},\ \bibinfo {pages} {010318} (\bibinfo {year} {2022})}\BibitemShut {NoStop}%
\bibitem [{\citenamefont {Childs}\ and\ \citenamefont {Li}(2016)}]{childs2016efficient}%
  \BibitemOpen
  \bibfield  {author} {\bibinfo {author} {\bibfnamefont {A.~M.}\ \bibnamefont {Childs}}\ and\ \bibinfo {author} {\bibfnamefont {T.}~\bibnamefont {Li}},\ }\bibfield  {title} {\bibinfo {title} {Efficient simulation of sparse markovian quantum dynamics},\ }\href@noop {} {\bibfield  {journal} {\bibinfo  {journal} {arXiv preprint arXiv:1611.05543}\ } (\bibinfo {year} {2016})}\BibitemShut {NoStop}%
\bibitem [{\citenamefont {Kitaev}\ and\ \citenamefont {Webb}(2008)}]{kitaev2008wavefunction}%
  \BibitemOpen
  \bibfield  {author} {\bibinfo {author} {\bibfnamefont {A.}~\bibnamefont {Kitaev}}\ and\ \bibinfo {author} {\bibfnamefont {W.~A.}\ \bibnamefont {Webb}},\ }\bibfield  {title} {\bibinfo {title} {Wavefunction preparation and resampling using a quantum computer},\ }\href@noop {} {\bibfield  {journal} {\bibinfo  {journal} {arXiv preprint arXiv:0801.0342}\ } (\bibinfo {year} {2008})}\BibitemShut {NoStop}%
\bibitem [{\citenamefont {Aaronson}\ and\ \citenamefont {Rall}(2020)}]{aaronson2020quantum}%
  \BibitemOpen
  \bibfield  {author} {\bibinfo {author} {\bibfnamefont {S.}~\bibnamefont {Aaronson}}\ and\ \bibinfo {author} {\bibfnamefont {P.}~\bibnamefont {Rall}},\ }\bibfield  {title} {\bibinfo {title} {Quantum approximate counting, simplified},\ }in\ \href@noop {} {\emph {\bibinfo {booktitle} {Symposium on simplicity in algorithms}}}\ (\bibinfo {organization} {SIAM},\ \bibinfo {year} {2020})\ pp.\ \bibinfo {pages} {24--32}\BibitemShut {NoStop}%
\bibitem [{\citenamefont {Venegas-Andraca}(2012)}]{venegas2012quantum}%
  \BibitemOpen
  \bibfield  {author} {\bibinfo {author} {\bibfnamefont {S.~E.}\ \bibnamefont {Venegas-Andraca}},\ }\bibfield  {title} {\bibinfo {title} {Quantum walks: a comprehensive review},\ }\href@noop {} {\bibfield  {journal} {\bibinfo  {journal} {Quantum Information Processing}\ }\textbf {\bibinfo {volume} {11}},\ \bibinfo {pages} {1015} (\bibinfo {year} {2012})}\BibitemShut {NoStop}%
\bibitem [{\citenamefont {Chowdhury}\ and\ \citenamefont {Somma}(2016)}]{chowdhury2016quantum}%
  \BibitemOpen
  \bibfield  {author} {\bibinfo {author} {\bibfnamefont {A.~N.}\ \bibnamefont {Chowdhury}}\ and\ \bibinfo {author} {\bibfnamefont {R.~D.}\ \bibnamefont {Somma}},\ }\bibfield  {title} {\bibinfo {title} {Quantum algorithms for {G}ibbs sampling and hitting-time estimation},\ }\href@noop {} {\bibfield  {journal} {\bibinfo  {journal} {arXiv preprint arXiv:1603.02940}\ } (\bibinfo {year} {2016})}\BibitemShut {NoStop}%
\bibitem [{\citenamefont {Sachdeva}\ and\ \citenamefont {Vishnoi}(2014)}]{sachdeva2014faster}%
  \BibitemOpen
  \bibfield  {author} {\bibinfo {author} {\bibfnamefont {S.}~\bibnamefont {Sachdeva}}\ and\ \bibinfo {author} {\bibfnamefont {N.~K.}\ \bibnamefont {Vishnoi}},\ }\bibfield  {title} {\bibinfo {title} {Faster algorithms via approximation theory},\ }\href@noop {} {\bibfield  {journal} {\bibinfo  {journal} {Foundations and Trends{\textregistered} in Theoretical Computer Science}\ }\textbf {\bibinfo {volume} {9}},\ \bibinfo {pages} {125} (\bibinfo {year} {2014})}\BibitemShut {NoStop}%
\bibitem [{\citenamefont {Lidar}\ and\ \citenamefont {Brun}(2013)}]{lidar2013quantum}%
  \BibitemOpen
  \bibfield  {author} {\bibinfo {author} {\bibfnamefont {D.~A.}\ \bibnamefont {Lidar}}\ and\ \bibinfo {author} {\bibfnamefont {T.~A.}\ \bibnamefont {Brun}},\ }\href@noop {} {\emph {\bibinfo {title} {Quantum error correction}}}\ (\bibinfo  {publisher} {Cambridge university press},\ \bibinfo {year} {2013})\BibitemShut {NoStop}%
\bibitem [{\citenamefont {Fagnola}\ \emph {et~al.}(2019)\citenamefont {Fagnola}, \citenamefont {Gough}, \citenamefont {Nurdin},\ and\ \citenamefont {Viola}}]{fagnola2019mathematical}%
  \BibitemOpen
  \bibfield  {author} {\bibinfo {author} {\bibfnamefont {F.}~\bibnamefont {Fagnola}}, \bibinfo {author} {\bibfnamefont {J.~E.}\ \bibnamefont {Gough}}, \bibinfo {author} {\bibfnamefont {H.~I.}\ \bibnamefont {Nurdin}},\ and\ \bibinfo {author} {\bibfnamefont {L.}~\bibnamefont {Viola}},\ }\bibfield  {title} {\bibinfo {title} {Mathematical models of markovian dephasing},\ }\href@noop {} {\bibfield  {journal} {\bibinfo  {journal} {Journal of Physics A: Mathematical and Theoretical}\ }\textbf {\bibinfo {volume} {52}},\ \bibinfo {pages} {385301} (\bibinfo {year} {2019})}\BibitemShut {NoStop}%
\bibitem [{\citenamefont {Krawtchouk}(1929)}]{krawtchouk1929generalisation}%
  \BibitemOpen
  \bibfield  {author} {\bibinfo {author} {\bibfnamefont {M.}~\bibnamefont {Krawtchouk}},\ }\bibfield  {title} {\bibinfo {title} {Sur une g{\'e}n{\'e}ralisation des polyn{\^o}mes d’hermite},\ }\href@noop {} {\bibfield  {journal} {\bibinfo  {journal} {Comptes Rendus}\ }\textbf {\bibinfo {volume} {189}},\ \bibinfo {pages} {5} (\bibinfo {year} {1929})}\BibitemShut {NoStop}%
\bibitem [{\citenamefont {Moroder}\ \emph {et~al.}(2024)\citenamefont {Moroder}, \citenamefont {Culhane}, \citenamefont {Zawadzki},\ and\ \citenamefont {Goold}}]{moroder2024thermodynamics}%
  \BibitemOpen
  \bibfield  {author} {\bibinfo {author} {\bibfnamefont {M.}~\bibnamefont {Moroder}}, \bibinfo {author} {\bibfnamefont {O.}~\bibnamefont {Culhane}}, \bibinfo {author} {\bibfnamefont {K.}~\bibnamefont {Zawadzki}},\ and\ \bibinfo {author} {\bibfnamefont {J.}~\bibnamefont {Goold}},\ }\bibfield  {title} {\bibinfo {title} {Thermodynamics of the quantum mpemba effect},\ }\href@noop {} {\bibfield  {journal} {\bibinfo  {journal} {Physical Review Letters}\ }\textbf {\bibinfo {volume} {133}},\ \bibinfo {pages} {140404} (\bibinfo {year} {2024})}\BibitemShut {NoStop}%
\bibitem [{\citenamefont {Vershynin}(2018)}]{Vershynin_2018}%
  \BibitemOpen
  \bibfield  {author} {\bibinfo {author} {\bibfnamefont {R.}~\bibnamefont {Vershynin}},\ }\href@noop {} {\emph {\bibinfo {title} {High-Dimensional Probability: An Introduction with Applications in Data Science}}},\ Cambridge Series in Statistical and Probabilistic Mathematics\ (\bibinfo  {publisher} {Cambridge University Press},\ \bibinfo {year} {2018})\BibitemShut {NoStop}%
\bibitem [{\citenamefont {Cuccaro}\ \emph {et~al.}(2004)\citenamefont {Cuccaro}, \citenamefont {Draper}, \citenamefont {Kutin},\ and\ \citenamefont {Moulton}}]{cuccaro2004new}%
  \BibitemOpen
  \bibfield  {author} {\bibinfo {author} {\bibfnamefont {S.~A.}\ \bibnamefont {Cuccaro}}, \bibinfo {author} {\bibfnamefont {T.~G.}\ \bibnamefont {Draper}}, \bibinfo {author} {\bibfnamefont {S.~A.}\ \bibnamefont {Kutin}},\ and\ \bibinfo {author} {\bibfnamefont {D.~P.}\ \bibnamefont {Moulton}},\ }\bibfield  {title} {\bibinfo {title} {A new quantum ripple-carry addition circuit},\ }\href@noop {} {\bibfield  {journal} {\bibinfo  {journal} {arXiv preprint quant-ph/0410184}\ } (\bibinfo {year} {2004})}\BibitemShut {NoStop}%
\bibitem [{\citenamefont {Draper}\ \emph {et~al.}(2004)\citenamefont {Draper}, \citenamefont {Kutin}, \citenamefont {Rains},\ and\ \citenamefont {Svore}}]{draper2004logarithmic}%
  \BibitemOpen
  \bibfield  {author} {\bibinfo {author} {\bibfnamefont {T.~G.}\ \bibnamefont {Draper}}, \bibinfo {author} {\bibfnamefont {S.~A.}\ \bibnamefont {Kutin}}, \bibinfo {author} {\bibfnamefont {E.~M.}\ \bibnamefont {Rains}},\ and\ \bibinfo {author} {\bibfnamefont {K.~M.}\ \bibnamefont {Svore}},\ }\bibfield  {title} {\bibinfo {title} {A logarithmic-depth quantum carry-lookahead adder},\ }\href@noop {} {\bibfield  {journal} {\bibinfo  {journal} {arXiv preprint quant-ph/0406142}\ } (\bibinfo {year} {2004})}\BibitemShut {NoStop}%
\bibitem [{\citenamefont {Bernard}\ and\ \citenamefont {Vinet}(2024)}]{bernard2024dynamical}%
  \BibitemOpen
  \bibfield  {author} {\bibinfo {author} {\bibfnamefont {P.-A.}\ \bibnamefont {Bernard}}\ and\ \bibinfo {author} {\bibfnamefont {L.}~\bibnamefont {Vinet}},\ }\bibfield  {title} {\bibinfo {title} {A dynamical algebra of protocol-induced transformations on dicke states},\ }\href@noop {} {\bibfield  {journal} {\bibinfo  {journal} {arXiv preprint arXiv:2412.17917}\ } (\bibinfo {year} {2024})}\BibitemShut {NoStop}%
\bibitem [{\citenamefont {Bravyi}\ \emph {et~al.}(2021)\citenamefont {Bravyi}, \citenamefont {Chowdhury}, \citenamefont {Gosset},\ and\ \citenamefont {Wocjan}}]{bravyi2021complexity}%
  \BibitemOpen
  \bibfield  {author} {\bibinfo {author} {\bibfnamefont {S.}~\bibnamefont {Bravyi}}, \bibinfo {author} {\bibfnamefont {A.}~\bibnamefont {Chowdhury}}, \bibinfo {author} {\bibfnamefont {D.}~\bibnamefont {Gosset}},\ and\ \bibinfo {author} {\bibfnamefont {P.}~\bibnamefont {Wocjan}},\ }\bibfield  {title} {\bibinfo {title} {On the complexity of quantum partition functions},\ }\href@noop {} {\bibfield  {journal} {\bibinfo  {journal} {arXiv preprint arXiv:2110.15466}\ } (\bibinfo {year} {2021})}\BibitemShut {NoStop}%
\bibitem [{\citenamefont {Bauer}\ \emph {et~al.}(2021)\citenamefont {Bauer}, \citenamefont {Deliyannis}, \citenamefont {Freytsis},\ and\ \citenamefont {Nachman}}]{bauer2021practical}%
  \BibitemOpen
  \bibfield  {author} {\bibinfo {author} {\bibfnamefont {C.~W.}\ \bibnamefont {Bauer}}, \bibinfo {author} {\bibfnamefont {P.}~\bibnamefont {Deliyannis}}, \bibinfo {author} {\bibfnamefont {M.}~\bibnamefont {Freytsis}},\ and\ \bibinfo {author} {\bibfnamefont {B.}~\bibnamefont {Nachman}},\ }\bibfield  {title} {\bibinfo {title} {Practical considerations for the preparation of multivariate gaussian states on quantum computers},\ }\href@noop {} {\bibfield  {journal} {\bibinfo  {journal} {arXiv preprint arXiv:2109.10918}\ } (\bibinfo {year} {2021})}\BibitemShut {NoStop}%
\bibitem [{\citenamefont {Petrov}(2012)}]{petrov2012sums}%
  \BibitemOpen
  \bibfield  {author} {\bibinfo {author} {\bibfnamefont {V.~V.}\ \bibnamefont {Petrov}},\ }\href@noop {} {\emph {\bibinfo {title} {Sums of independent random variables}}},\ Vol.~\bibinfo {volume} {82}\ (\bibinfo  {publisher} {Springer Science \& Business Media},\ \bibinfo {year} {2012})\BibitemShut {NoStop}%
\bibitem [{\citenamefont {Bagherimehrab}\ \emph {et~al.}(2022)\citenamefont {Bagherimehrab}, \citenamefont {Sanders}, \citenamefont {Berry}, \citenamefont {Brennen},\ and\ \citenamefont {Sanders}}]{bagherimehrab2022nearly}%
  \BibitemOpen
  \bibfield  {author} {\bibinfo {author} {\bibfnamefont {M.}~\bibnamefont {Bagherimehrab}}, \bibinfo {author} {\bibfnamefont {Y.~R.}\ \bibnamefont {Sanders}}, \bibinfo {author} {\bibfnamefont {D.~W.}\ \bibnamefont {Berry}}, \bibinfo {author} {\bibfnamefont {G.~K.}\ \bibnamefont {Brennen}},\ and\ \bibinfo {author} {\bibfnamefont {B.~C.}\ \bibnamefont {Sanders}},\ }\bibfield  {title} {\bibinfo {title} {Nearly optimal quantum algorithm for generating the ground state of a free quantum field theory},\ }\href@noop {} {\bibfield  {journal} {\bibinfo  {journal} {PRX Quantum}\ }\textbf {\bibinfo {volume} {3}},\ \bibinfo {pages} {020364} (\bibinfo {year} {2022})}\BibitemShut {NoStop}%
\end{thebibliography}%
\clearpage

\begin{appendix}
\renewcommand\thefigure{\thesection.\arabic{figure}}
\onecolumngrid
\renewcommand{\addcontentsline}{\oldacl}
\renewcommand{\tocname}{Appendix Contents}
\tableofcontents
\section{Setup and preliminaries\label{sma}}
In this work, the $n$-qubit Hamiltonian we consider has the form
\begin{eqnarray}\label{hami}
H=\sum_\alpha h_\alpha \Pi_\alpha,
\end{eqnarray}
where $\{h_\alpha\}\in[0,1]$ are real eigenvalues satisfying $h_\alpha\neq h_\beta,~\forall \alpha\neq \beta$, and $\{\Pi_\alpha\}$ are the corresponding eigenspace projectors.

We also introduce two concentration results here.


\begin{lemma}\label{con1}
Given a binomial distribution $\bm{\Pr}[X=m]=\binom{N}{m}(1-p)^{N-m}p^m$ and a positive $c$, we have the following inequality
\begin{eqnarray}
\bm{\Pr}[|X-Np|\geq cN]=\sum_{|m-Np|\geq cp}\binom{N}{m}(1-p)^{N-m}p^m\leq 2\exp\left(-\frac{Nc^2}{0.5p(1-p)+2c/3}\right).
\end{eqnarray}
\end{lemma}
\begin{proof}
Bernstein's inequality \cite{Vershynin_2018} states that when $X_1,\ldots,X_N$ are i.i.d.~variables with the mean $\mu$ and the variance $\sigma^2$, and we have $|X_i-\mu|\leq R$ for any $i$, then the following inequality holds
\begin{eqnarray}\label{bernst}
\bm{\Pr}\left[\left|\sum_i X_i-N\mu\right|\geq Nc\right]\leq 2\exp\left(\frac{-Nc^2}{2\sigma^2+2Rc/3}\right).
\end{eqnarray}
Note that the binomial distribution can be generated from $N$ independent Bernoulli coins with a probability $p$ to get 1 and a probability $1-p$ to get 0. Therefore, we can simply let $X_i$ be these Bernoulli coins, resulting in $\mu=p$, $\sigma^2=p(1-p)/4$ and $R=1$. Putting these parameters into Eq. \ref{bernst}, the result is proved.
\end{proof}

\begin{lemma}\label{con2}
Given a binomial distribution $\bm{\Pr}[X=m]=\binom{N}{m}(1-p)^{N-m}p^m$ and a positive $c$, we have the following inequality
\begin{eqnarray}
\bm{\Pr}[|X-Np|\geq cN]=\sum_{|m-Np|\geq cp}\binom{N}{m}(1-p)^{N-m}p^m\leq 2e^{-2c^2N}.
\end{eqnarray}
\end{lemma}
\begin{proof}
Here, the result is based on the Hoeffding inequality, which for example, can be found in Ref. \cite{Vershynin_2018}.
\end{proof}

\section{Brief review of standard quantum phase estimation (QPE)\label{qpeintro}}

\begin{figure}[htbp]
\centering
\includegraphics[width=0.9\textwidth]{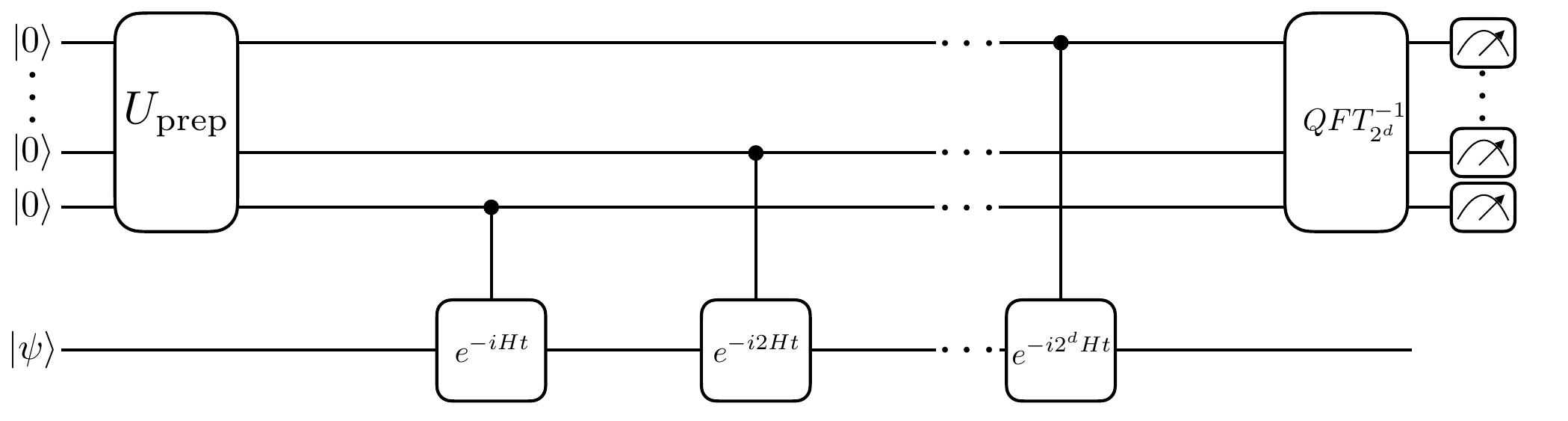}
\caption{The quantum circuit of standard quantum phase estimation \cite{nielsen2010quantum}.\label{fig.qpe}}
\end{figure}

Here, we summarize the results of standard QPE that utilizes the quantum Fourier transform (QFT) \cite{nielsen2010quantum}.

Given a state $|\psi\rangle=\sum_\alpha \Pi_\alpha |\psi\rangle=\sum_\alpha c_\alpha |\psi_\alpha\rangle$, QPE with $d$ ancilla qubits first acts a controlled Hamiltonian simulation unitary $U_H$ with the form
\begin{eqnarray}\label{stanhs}
U_H=\sum_{x=0}^{2^d-1} |x\rangle\langle x|\otimes e^{-i 2\pi H x},
\end{eqnarray}
where $x=(x_{d-1}x_{d-2}\cdots x_0)=\sum_{j=0}^{d-1} x_j 2^j$ on the initial state $|\text{in}\rangle=\sum_{x=0}^{2^d-1}\frac{1}{\sqrt{2^d}}|x\rangle\otimes |\psi\rangle$. Next, QPE implements the QFT circuit $U_{\text{QFT}}$ on the ancilla qubits. The output state can then be derived
\begin{eqnarray}\label{qpeout}
|\text{out}\rangle\ &&=(U_{\text{QFT}}\otimes I_n)U_H\sum_{x=0}^{2^d-1}\frac{1}{\sqrt{2^d}}|x\rangle\otimes |\psi\rangle\nonumber\\
&&=\sum_{x=0}^{2^d-1}\frac{1}{\sqrt{2^d}}U_{\text{QFT}}|x\rangle\otimes e^{-i 2\pi H x}|\psi\rangle\nonumber\\&&
=\sum_{x=0}^{2^d-1}\frac{1}{\sqrt{2^d}}U_{\text{QFT}}|x\rangle\otimes\sum_\alpha c_\alpha  e^{-i2\pi h_\alpha x}|\psi_\alpha\rangle\nonumber\\&&
=\sum_{x=0}^{2^d-1}\frac{1}{2^d}\sum_{y=0}^{2^d-1}e^{i 2\pi xy/2^d}|y\rangle\otimes\sum_\alpha c_\alpha  e^{-i2\pi h_\alpha x}|\psi_\alpha\rangle\nonumber\\&&
=\frac{1}{2^d}\sum_\alpha c_\alpha \sum_{y=0}^{2^d-1}\sum_{x=0}^{2^d-1}e^{-i 2\pi (h_\alpha -y/2^d)x}|y\rangle\otimes|\psi_\alpha\rangle\nonumber\\&&
=\frac{1}{2^d}\sum_\alpha c_\alpha \sum_{y=0}^{2^d-1}\frac{1-e^{-i 2\pi (h_\alpha -y/2^d)2^d}}{1-e^{-i 2\pi (h_\alpha -y/2^d)}}|y\rangle\otimes|\psi_\alpha\rangle.
\end{eqnarray}

\textbf{Eigenvalue estimation.——}For the eigenvalue estimation task, we can see that when $|\psi\rangle=|\psi_\alpha\rangle$ is an eigenstate of $H$, then Eq. \ref{qpeout} has the simplified form
\begin{eqnarray}
|\text{out}\rangle=\frac{1}{2^d}\sum_{y=0}^{2^d-1}\frac{1-e^{-i 2\pi (h_\alpha -y/2^d)2^d}}{1-e^{-i 2\pi (h_\alpha -y/2^d)}}|y\rangle\otimes|\psi_\alpha\rangle.
\end{eqnarray}
The probability of measuring $|y\rangle$ is therefore
\begin{eqnarray}\label{py1}
p_y\ &&=\frac{1}{2^{2d}}\frac{1-e^{-i 2\pi (h_\alpha -y/2^d)2^d}}{1-e^{-i 2\pi (h_\alpha -y/2^d)}}\frac{1-e^{i 2\pi (h_\alpha -y/2^d)2^d}}{1-e^{i 2\pi (h_\alpha -y/2^d)}}\nonumber\\
&&=\frac{1}{2^{2d}}\frac{ e^{-i \pi (h_\alpha-y/2^d)2^d}\sin(\pi (h_\alpha-y/2^d)2^d)}{ e^{-i \pi (h_\alpha-y/2^d)}\sin(\pi (h_\alpha-y/2^d))}\frac{ e^{i \pi (h_\alpha-y/2^d)2^d}\sin(\pi (h_\alpha-y/2^d)2^d)}{ e^{i \pi (h_\alpha-y/2^d)}\sin(\pi (h_\alpha-y/2^d))}\nonumber\\
&&=\frac{1}{2^{2d}}\frac{\sin(\pi (h_\alpha-y/2^d)2^d)^2}{ \sin(\pi (h_\alpha-y/2^d))^2}\nonumber\\
&&\leq \frac{1}{2^{2d}}\frac{1}{ \sin(\pi (h_\alpha-y/2^d))^2}\nonumber\\
&&=\frac{1}{2^{2d}}\frac{1}{ \sin(\pi |h_\alpha-y/2^d|_c)^2}\leq \frac{1}{2^{2d}}\frac{1}{ 4|h_\alpha-y/2^d|_c^2},
\end{eqnarray}
where the last line, we define the circular distance $\{|\theta|_c=\min\{\theta\mod 1, (1-\theta)\mod1\}$, and we use the inequality
\begin{eqnarray}
|\sin(\pi|\theta|_c)|\geq 2|\theta|_c.
\end{eqnarray}
Thus, we have the failure probability of estimation within an error $\epsilon$ has the upper bound
\begin{eqnarray}
\bm{\Pr}\left[|h_\alpha-y/2^d|_c\geq \epsilon\right]&&=\sum_{y: |h_\alpha-y/2^d|_c\geq \epsilon} p_y\\
&&\leq \sum_{y: |h_\alpha-y/2^d|_c\geq \epsilon}\frac{1}{2^{2d}}\frac{1}{ 4|h_\alpha-y/2^d|_c^2}\nonumber\\
&&\leq\frac{1}{2\cdot 2^d} \int_{\epsilon}^\infty \frac{1}{s^2}ds+\frac{1}{2\cdot 2^{2d}\epsilon^2}=\frac{1}{2\cdot 2^{d}\epsilon}+\frac{1}{2\cdot 2^{2d}\epsilon^2}\nonumber\\&&\leq \frac{1}{2^{d}\epsilon}.
\end{eqnarray}
As a result, to have a success probability at least $1-\delta$, we require $2^d=\mathcal{O}(1/(\epsilon\delta))$. Since $U_H$ has the decomposition 
\begin{eqnarray}\label{uh}
U_H=\sum_{x=0}^{2^d-1} |x\rangle\langle x|\otimes e^{-i 2\pi H x}=\prod_{j=0}^{d-1}\left(|0\rangle\langle 0|_j\otimes I+|1\rangle\langle 1|_j\otimes e^{-i2\pi H 2^{j}}\right),
\end{eqnarray}
which is the successive multiplication of $j$th-ancilla qubit controlled $e^{-i2\pi H 2^{j}}$ unitaries for $j\in\{0,1,\ldots,d-1\}$, therefore, the cost of constructing $U_H$ characterized by the Hamiltonian simulation time is $1+2+\cdots+2^{d-1}=2^d-1$ and we have
\begin{eqnarray}
\text{Hami.simu.time}=2^d=
\mathcal{O}(\epsilon^{-1}\delta^{-1}),
\end{eqnarray}
where the $\epsilon^{-1}$ scaling is also called the Heisenberg limit.

\textbf{Eigenstate preparation.——}For the eigenstate preparation task, we still consider Eq. \ref{qpeout}. Following the derivation of Eq. \ref{py1}, the probability of measuring $y$ on the ancilla register is 
\begin{eqnarray}\label{py2}
p_y&&=\frac{1}{2^{2d}}\sum_\alpha |c_\alpha|^2\frac{\sin(\pi (h_\alpha-y/2^d)2^d)^2}{ \sin(\pi (h_\alpha-y/2^d))^2}.
\end{eqnarray}
After the measurement, the remaining state on the system register becomes
\begin{eqnarray}
|\phi_y\rangle=\frac{1}{2^d\sqrt{p_y}}\sum_\alpha c_\alpha \frac{1-e^{-i 2\pi (h_\alpha -y/2^d)2^d}}{1-e^{-i 2\pi (h_\alpha -y/2^d)}}|\psi_\alpha\rangle,
\end{eqnarray}
whose overlap with an eigenstate $|\psi_\beta\rangle$ is
\begin{eqnarray}\label{pya}
|\langle\phi_y|\psi_\beta\rangle|^2&&=\frac{|c_\beta|^2\frac{\sin(\pi (h_\beta-y/2^d)2^d)^2}{ \sin(\pi (h_\beta-y/2^d))^2}}{2^{2d}p_y}=\frac{|c_\beta|^2\frac{\sin(\pi (h_\beta-y/2^d)2^d)^2}{ \sin(\pi (h_\beta-y/2^d))^2}}{|c_\beta|^2\frac{\sin(\pi (h_\beta-y/2^d)2^d)^2}{ \sin(\pi (h_\beta-y/2^d))^2}+\sum_{\alpha\neq\beta} |c_\alpha|^2\frac{\sin(\pi (h_\alpha-y/2^d)2^d)^2}{ \sin(\pi (h_\alpha-y/2^d))^2}}.
\end{eqnarray}
If $h_\beta$ is exactly known, we can then always let $h_\beta$ be 0 by shifting the spectrum of $H$ (adding identity). In this case, we will have
\begin{eqnarray}
|\langle\phi_y|\psi_\beta\rangle|^2=0,~\forall y\neq 0.
\end{eqnarray}
Therefore, only when the measurement outcome is $y=0$, the eigenstate $|\psi_\beta\rangle$ will have a non-zero contribution. For $y=0$, we have the measurement probability
\begin{eqnarray}
p_0&&=|c_\beta|^2+\frac{1}{2^{2d}}\sum_{\alpha\neq\beta} |c_\alpha|^2\frac{\sin(\pi h_\alpha2^d)^2}{ \sin(\pi h_\alpha)^2}\geq |c_\beta|^2.
\end{eqnarray}
Define $\Delta_\beta$ as the gap between $h_\beta$ and other eigenvalues, then the overlap has a lower bound
\begin{eqnarray}
|\langle\phi_0|\psi_\beta\rangle|^2&&=\frac{|c_\beta|^22^{2d}}{|c_\beta|^22^{2d}+\sum_{\alpha\neq\beta} |c_\alpha|^2\frac{\sin(\pi (h_\alpha-y/2^d)2^d)^2}{ \sin(\pi (h_\alpha-y/2^d))^2}}\nonumber\\
&&\geq \frac{|c_\beta|^22^{2d}}{|c_\beta|^22^{2d}+\sum_{\alpha\neq\beta} |c_\alpha|^2\frac{1}{ 4|h_\alpha|_c^2}}\nonumber\\
&&\geq \frac{|c_\beta|^22^{2d}}{|c_\beta|^22^{2d}+\sum_{\alpha\neq\beta} |c_\alpha|^2\frac{1}{ 4\Delta_\beta^2}}\nonumber\\
&&=\frac{|c_\beta|^2}{|c_\beta|^2+(1-|c_\beta|^2)\frac{1}{ 4(2^{d}\Delta_\beta)^2}}.
\end{eqnarray}
To have an overlap at least $1-\zeta$, we require $2^d=\mathcal{O}(|c_\beta|^{-1}\Delta_\beta^{-1}\zeta^{-1/2})$. For the overall complexity, we need to account for the measurement probability $p_0$. Through amplitude amplification, we can use $\mathcal{O}(|c_\beta|^{-1})$ queries on the preparation unitaries of $|\text{out}\rangle$ to amplify $p_0$ to $\mathcal{O}(1)$. As a result, the overall required Hamiltonian simulation time is
\begin{eqnarray}
\text{Hami.simu.time}=\mathcal{O}(|c_\beta|^{-1}2^d)=\mathcal{O}(|c_\beta|^{-2}\Delta_\beta^{-1}\zeta^{-1/2}).
\end{eqnarray}

\section{Dilated Hamiltonian for Lindbladian simulation \label{smc}}
The general form of a Lindbladian is
\begin{equation}\label{lme}
\frac{\d\rho(t)}{\d t}=\mathcal{L}[\rho(t)]=-i[H_I,\rho(t)]+
\sum_i \left(F_i\rho(t) F_i^\dag-
\frac{1}{2}\{\rho(t),F_i^\dag F_i\}\right),
\end{equation}
where $\rho(t)$ is system density matrix, $H_I$ is the internal Hamiltonian, and $F_i$ are quantum jump operators. Here, we consider a special case where $H_I=0$ and there is only a single jump operator $F$
\begin{equation}\label{flme}
\frac{d\rho(t)}{dt}=\mathcal{L}_F[\rho(t)]=F\rho(t) F^\dag-
\frac{1}{2}\{\rho(t),F^\dag F\}.
\end{equation}

For this special Lindbladian Eq. \ref{flme}, there is a simple Lindbladian simulation algorithm through dilated Hamiltonian
\begin{equation}
\tilde{F}=\begin{pmatrix}
0 & F^\dag \\
F & 0
\end{pmatrix}.
\end{equation}
The short-time evolution under $\tilde{F}$ has the effect
\begin{eqnarray}\label{dil}
e^{-i \tilde{F}\sqrt{\tau}}(|0\rangle\langle 0|\otimes \rho(t))e^{i \tilde{F}\sqrt{\tau}}&&\approx \left(I-i \tilde{F}\sqrt{\tau}-\frac{1}{2}\tilde{F}^2\tau\right)(|0\rangle\langle 0|\otimes \rho(t))\left(I+i \tilde{F}\sqrt{\tau}-\frac{1}{2}\tilde{F}^2\tau\right)\nonumber\\&&\approx |0\rangle\langle 0|\otimes \rho(t)+ \tau\tilde{F}(|0\rangle\langle 0|\otimes \rho(t))\tilde{F}\nonumber-\frac{1}{2}\tau\tilde{F}^2(|0\rangle\langle 0|\otimes \rho(t))-\frac{1}{2}\tau(|0\rangle\langle 0|\otimes \rho(t))\tilde{F}^2\nonumber\\&&=|0\rangle\langle 0|\otimes \left(\rho(t)-\frac{1}{2}\{F^\dag F,\rho(t)\}\tau\right)+|1\rangle\langle 1|\otimes F\rho(t)F^\dag\tau.
\end{eqnarray}
By tracing out the ancilla qubit, we obtain
\begin{eqnarray}
\rho(t)+F\rho(t)F^\dag\tau-\frac{1}{2}\{F^\dag F,\rho(t)\}\tau\approx \rho(t+\tau).
\end{eqnarray}
Therefore, by repeatedly adding ancilla qubits and running short-time evolution of $\tilde{F}$, we are able to simulate Eq. \ref{flme}. 

Regarding the complexity, for a simulation of time $\tau=t/N$, dilated Hamiltonian simulation needs a time $\sqrt{\tau}=\sqrt{t/N}$, resulting in a total simulation time $T=N\sqrt{t/N}=\sqrt{Nt}$. It can be shown that to simulate the Lindbladian Eq. \ref{flme} for time $t$ within error $\varepsilon$, we need $N=t^3/\varepsilon^2$, resulting in a total dilated Hamiltonian simulation time
\begin{eqnarray}\label{dilated}
T=\mathcal{O}\left(\frac{t^2}{\varepsilon}\right).
\end{eqnarray}
Since $N$ is also the number of ancilla qubits, the number of ancilla qubits is $\mathcal{O}(t^3\varepsilon^{-2})$.

\section{Details on Lindbladian as slow QPE}
\subsection{Motivation\label{smmoti}}
Consider the following special Lindbladian
\begin{equation}\label{slme}
\frac{\d\rho(t)}{\d t}=\mathcal{L}_H[\rho(t)]=H\rho(t) H-
\frac{1}{2}\{\rho(t),H^2\},
\end{equation}
where the single jump operator is the Hamiltonian $H$ in Eq. \ref{hami}. We can express $\rho(0)$ in terms of $\{\Pi_\alpha\}$:
\begin{equation}
\rho(0)=\sum_{\alpha\beta}\Pi_\alpha\rho(0)\Pi_\beta=\sum_{\alpha\beta}\rho_{0,\alpha\beta}.
\end{equation}
Since we have
\begin{eqnarray}
&&\frac{\d\Pi_\alpha\rho(0)\Pi_\beta}{\d t}=h_\alpha h_\beta\rho_{0,\alpha\beta}-\frac{h_\alpha^2+h_\beta^2}{2}\rho_{0,\alpha\beta}\\&&
\rho_{0,\alpha\beta}\xrightarrow{t}e^{\left(h_\alpha h_\beta-\frac{h_\alpha^2+h_\beta^2}{2}\right)t}\rho_{0,\alpha\beta},
\end{eqnarray}
Therefore, the Lindbladian Eq. \ref{slme} has the solution
\begin{eqnarray}
\rho(t)&&=e^{\mathcal{L}_Ht}[\rho(0)]=\sum_{\alpha\beta} e^{\mathcal{L}_Ht}[\rho_{0,\alpha\beta}]=\sum_{\alpha\beta} e^{\left(h_\alpha h_\beta-\frac{h_\alpha^2+h_\beta^2}{2}\right)t}\rho_{0,\alpha\beta}.
\end{eqnarray}
Because we have $h_\alpha\neq h_\beta$ for $\alpha\neq \beta$, therefore, 
\begin{eqnarray}
h_\alpha h_\beta-\frac{h_\alpha^2+h_\beta^2}{2}&&\leq |h_\alpha||h_\beta|-\frac{h_\alpha^2+h_\beta^2}{2}< 0,~\forall\,\alpha\neq \beta,\\
h_ \alpha h_\beta-\frac{h_\alpha^2+h_\beta^2}{2}&&=0,~\forall\,\alpha= \beta.
\end{eqnarray}
Thus, the steady state $\rho_{\text{ss}}$ preserves the coherence inside the subspace of $H$ with the same energy while dissipating the coherence between the subspace of $H$ with different energy
\begin{eqnarray}\label{steady}
\rho_{\text{ss}}=\rho(\infty)=\sum_{\alpha}\rho_{0,\alpha\alpha}.
\end{eqnarray}

For QPE as shown in Section \ref{qpeintro}, considering infinite accuracy, the output state, including both ancilla and the system, has the form
\begin{eqnarray}
\rho_{\text{QPE}}=\sum_{\alpha\beta}|h_\alpha\rangle\langle h_\beta|\otimes \rho_{0,\alpha\beta},
\end{eqnarray}
where $|h_\alpha\rangle$ is an ancilla computation basis state corresponding to the binary representation of the eigenvalue $h_\alpha$. By tracing out the ancillas, we can find that the reduced system Hamiltonian is also $\rho_{\text{ss}}$, which indicates a strong connection between these two seemingly vastly different quantum methods.

\subsection{Lindbladian as slow QPE\label{smslow}}
We consider in more detail the short-time dilated Hamiltonian simulation Eq. \ref{dil}, we have
\begin{eqnarray}
e^{-i \tilde{F}\sqrt{\tau}}&&=\sum_{k=0}^\infty\frac{(-i\sqrt{\tau})^k}{k!}\tilde{F}^k=\sum_{k=0}^\infty \frac{(-i\sqrt{\tau})^{2k}}{2k!}\tilde{F}^{2k}+\sum_{k=0}^\infty \frac{(-i\sqrt{\tau})^{(2k+1)}}{(2k+1)!}\tilde{F}^{2k+1}\nonumber\\&&=\sum_{k=0}^\infty \frac{(-i\sqrt{\tau})^{2k}}{2k!}\begin{pmatrix}
(F^\dag F)^k &0  \\
0 &  (F F^\dag)^k
\end{pmatrix}+\sum_{k=0}^\infty \frac{(-i\sqrt{\tau})^{(2k+1)}}{(2k+1)!}\begin{pmatrix}
0 &(F^\dag F)^k F^\dag\\
(F F^\dag)^k F &  0
\end{pmatrix}.
\end{eqnarray}
$F$ has the singular value decomposition $F=W\Sigma V$, putting this into the above expression, we have
\begin{eqnarray}\label{short}
e^{-i \tilde{F}\sqrt{\tau}}&&=\sum_{k=0}^\infty \frac{(-i\sqrt{\tau})^{2k}}{2k!}\begin{pmatrix}
V^\dag\Sigma^{2k}V &0  \\
0 &  W\Sigma^{2k}W^\dag
\end{pmatrix}+\sum_{k=0}^\infty \frac{(-i\sqrt{\tau})^{(2k+1)}}{(2k+1)!}\begin{pmatrix}
0 &V^\dag\Sigma^{2k+1} W^\dag\\
W\Sigma^{2k+1} V &  0\
\end{pmatrix}\nonumber\\&&=\begin{pmatrix}
V^\dag\cos(\sqrt{\tau}\Sigma)V &0  \\
0 &  W\cos(\sqrt{\tau}\Sigma)W^\dag
\end{pmatrix}-i\begin{pmatrix}
0 &V^\dag\sin(\sqrt{\tau}\Sigma) W^\dag\\
W\sin(\sqrt{\tau}\Sigma) V &  0\
\end{pmatrix}.
\end{eqnarray}
When $F=H$ is Hermitian, we have
\begin{eqnarray}\label{decom}
e^{-i \tilde{H}\sqrt{\tau}}=\begin{pmatrix}
\cos(\sqrt{\tau}H) &0  \\
0 &  \cos(\sqrt{\tau}H)
\end{pmatrix}-i\begin{pmatrix}
0 &\sin(\sqrt{\tau}H)\\
\sin(\sqrt{\tau}H) &  0\
\end{pmatrix}.
\end{eqnarray}

To simplify the presentation without loss of generality, we consider $\rho(0)=|\psi\rangle\langle \psi|$, which has the form
\begin{eqnarray}
|\psi\rangle=\sum_\alpha \Pi_\alpha |\psi\rangle=\sum_\alpha c_\alpha |\psi_\alpha\rangle,
\end{eqnarray}
with $\sum_i |c_i|^2=1$. Acting Eq. \ref{decom} to $|0\rangle|\psi\rangle$, we have
\begin{eqnarray}\label{shorth}
e^{-i \tilde{H}\sqrt{\tau}}|0\rangle|\psi\rangle=|0\rangle\cos(\sqrt{\tau}H)|\psi\rangle-i|1\rangle\sin(\sqrt{\tau}H)|\psi\rangle.
\end{eqnarray}
By setting $\tau=t/N$ and repeatedly implementing $e^{-i \tilde{H}}$ for $N$ times (with each time adding a new ancilla qubit), we have
\begin{eqnarray}\label{lint2}
\left(\prod_i e^{-i\tilde{H}_i\sqrt{\frac{t}{N}}}\right)|0\rangle^{\otimes N}|\psi\rangle&&=\sum_{m=0}^N (-i)^m\sqrt{\binom{N}{m}}|\bar{m}\rangle\cos\left(\sqrt{\frac{t}{N}}H\right)^{N-m}\sin\left(\sqrt{\frac{t}{N}}H\right)^{m}|\psi\rangle\\&&=\sum_{m=0}^N |\bar{m}\rangle(-i)^m\sqrt{\binom{N}{m}}\sum_\alpha c_\alpha\cos\left(\sqrt{\frac{t}{N}}h_\alpha\right)^{N-m}\sin\left(\sqrt{\frac{t}{N}}h_\alpha\right)^{m}|\psi_\alpha\rangle\nonumber,
\end{eqnarray}
where $|\bar{m}\rangle$ represents the Dicke state, the equal superposition of all $N-m$ 0s and $m$ 1s configurations.

The probability of measuring $|\bar{m}\rangle$ is
\begin{eqnarray}\label{pm}
p_m=\binom{N}{m}\sum_\alpha |c_\alpha|^2\cos\left(\sqrt{\frac{t}{N}}h_\alpha\right)^{2N-2m}\sin\left(\sqrt{\frac{t}{N}}h_\alpha\right)^{2m},
\end{eqnarray}
with the resulting state in the system having the form
\begin{eqnarray}
|\phi_m\rangle=\frac{1}{\sqrt{p_m}}(-i)^m\sqrt{\binom{N}{m}}\sum_\alpha c_\alpha\cos\left(\sqrt{\frac{t}{N}}h_\alpha\right)^{N-m}\sin\left(\sqrt{\frac{t}{N}}h_\alpha\right)^{m}|\psi_\alpha\rangle.
\end{eqnarray}
the contribution (probability) of $|\psi_\alpha\rangle$ in $|\phi_m\rangle$ is
\begin{eqnarray}\label{pma}
p_{m,\alpha}=|\langle\phi_m|\psi_\alpha\rangle|^2=\binom{N}{m}\frac{|c_\alpha|^2\cos\left(\sqrt{\frac{t}{N}}h_\alpha\right)^{2N-2m}\sin\left(\sqrt{\frac{t}{N}}h_\alpha\right)^{2m}}{p_m}.
\end{eqnarray}

\subsubsection{Eigenvalue estimation}

When $|\psi\rangle=|\psi_\alpha\rangle$ is an eigenstate of $H$, we have
\begin{eqnarray}
p_m&&=\binom{N}{m}\cos\left(\sqrt{\frac{t}{N}}h_\alpha\right)^{2N-2m}\sin\left(\sqrt{\frac{t}{N}}h_\alpha\right)^{2m}\nonumber\\
&&=\binom{N}{m}(1-q_\alpha)^{N-m}q_\alpha^m,
\end{eqnarray}
where we define $q_\alpha=\sin\left(\sqrt{\frac{t}{N}}h_\alpha\right)^{2}$. Therefore, $p_m$ is exactly the binomial distribution. According to Lemma \ref{con1}, the measurement results will concentrate on $q_\alpha N$ and we have
\begin{eqnarray}
\bm{\Pr}[|m-Nq_\alpha|\geq cN]=\sum_{|m-Nq_\alpha|\geq cN}p_m\leq 2\exp\left(-\frac{Nc^2}{0.5q_\alpha(1-q_\alpha)+2c/3}\right).
\end{eqnarray}
Therefore, we can expect to use $m/N$ to estimate $q_\alpha$. Since $\sqrt{\frac{t}{N}}\ll1$, we have $q_\alpha\approx \frac{t}{N}h_\alpha^2$. Therefore, an $\epsilon$-additive error estimation in $h_\alpha$ corresponds to a $\frac{2t}{N}h_\alpha\epsilon$-additive error estimation in $q_\alpha$, which further corresponds to a $2th_\alpha\epsilon$-uncertainty ($c=\frac{2th_\alpha}{N}\epsilon$) in $m$ and we have
\begin{eqnarray}
\bm{\Pr}[|m-Nq_\alpha|\geq 2th_\alpha\epsilon]&&\lesssim 2\exp\left(-\frac{4t^2h_\alpha^2\epsilon^2/N}{0.5\frac{t}{N}h_\alpha^2(1-\frac{t}{N}h_\alpha^2)+\frac{4th_\alpha\epsilon}{3N}}\right)\nonumber\\&&=
2\exp\left(-\frac{4t^2h_\alpha^2\epsilon^2}{0.5th_\alpha^2(1-\frac{t}{N}h_\alpha^2)+\frac{4th_\alpha\epsilon}{3}}\right)\nonumber\\&&\leq 2\exp\left(-\frac{4t^2h_\alpha^2\epsilon^2}{0.5th_\alpha^2+\frac{4th_\alpha\epsilon}{3}}\right)\nonumber\\
&&= 2\exp\left(-\frac{4th_\alpha\epsilon^2}{0.5h_\alpha+\frac{4\epsilon}{3}}\right).
\end{eqnarray}
We typically require an accuracy $\epsilon<h_\alpha$, which leads to
\begin{eqnarray}\label{c1}
\bm{\Pr}[|m-Nq_\alpha|\geq 2th_\alpha\epsilon]=\sum_{|m-Nq_\alpha|\geq 2th_\alpha\epsilon}p_m\leq 2\exp\left(-\frac{24}{11}t\epsilon^2\right).
\end{eqnarray}
Therefore, to have an estimation success probability at least $1-\delta$, we require the Lindbladian evolution time to be of order
\begin{eqnarray}\label{eslin}
\text{Lindlad.evo.time (t)}=\mathcal{O}(\epsilon^{-2}\log (\delta^{-1})).
\end{eqnarray}
Since the Lindbladian dynamics of time $t$ within an error $\varepsilon$ is simulated  through the dilated Hamiltonian with the simulation time to be $\mathcal{O}(t^2/\varepsilon)$ as shown in Eq. \ref{dilated}, therefore, the total dilated Hamiltonian simulation time required to fulfill the eigenvalue estimation task is
\begin{eqnarray}
\text{Hami.simu.time (T)}=\mathcal{O}(\epsilon^{-4}\log^2 (\delta^{-1})\varepsilon^{-1}),
\end{eqnarray}
where the $\varepsilon$ dependence just affects the Lindbladian simulation error but will not contribute to the accuracy of the eigenvalue estimation. Therefore, the natural (purified) evolution of the Lindbladian Eq. \ref{slme} gives an eigenvalue estimation algorithm with standard quantum limit ($\epsilon^{-2}$), and the dilated Hamiltonian simulation algorithm for Eq. \ref{slme} gives an eigenvalue estimation algorithm that is quadratically worse ($\epsilon^{-4}$) than standard quantum limit.

\subsubsection{Eigenstate preparation\label{slowep}}
Given $|\psi\rangle=\sum_\alpha c_\alpha|\psi_\alpha\rangle$, if the eigenvalue $h_\beta$ of the interested eigenstate $|\psi_\beta\rangle$ is exactly known, then similar to the operation in standard QPE (Section \ref{qpeintro}), we can then always let $h_\beta$ be 0 by shifting the spectrum of $H$. By this, following Eq. \ref{pma}, we will have
\begin{eqnarray}
p_{m,\beta}=\binom{N}{m}\frac{|c_\beta|^2\cos\left(\sqrt{\frac{t}{N}}h_\beta\right)^{2N-2m}\sin\left(\sqrt{\frac{t}{N}}h_\beta\right)^{2m}}{p_m}=0,~\forall m\neq0.
\end{eqnarray}
Therefore, we only need to focus on the measurement event $m=0$. According to Eq. \ref{pm}, the probability of measuring $m=0$ is
\begin{eqnarray}\label{import1}
p_0&&=\sum_\alpha |c_\alpha|^2\cos\left(\sqrt{\frac{t}{N}}h_\alpha\right)^{2N}\nonumber\\
&&=|c_\beta|^2+\sum_{\alpha\neq\beta} |c_\alpha|^2\cos\left(\sqrt{\frac{t}{N}}h_\alpha\right)^{2N}\geq |c_\beta|^2.
\end{eqnarray}
Let $\Delta_\beta$ be the spectral gap between $h_\beta$ and other eigenvalues, then after the measurement, the overlap between the resulting state and $|\psi_\beta\rangle$ has the lower bound
\begin{eqnarray}\label{import2}
p_{0,\beta}&&=\frac{|c_\beta|^2}{p_0}=\frac{|c_\beta|^2}{|c_\beta|^2+\sum_{\alpha\neq\beta} |c_\alpha|^2\cos\left(\sqrt{\frac{t}{N}}h_\alpha\right)^{2N}}\nonumber\\
&&\approx\frac{|c_\beta|^2}{|c_\beta|^2+\sum_{\alpha\neq\beta} |c_\alpha|^2(1-\frac{t}{2N}h_\alpha^2)^{2N}}\nonumber\\
&&\geq \frac{|c_\beta|^2}{|c_\beta|^2+\sum_{\alpha\neq\beta} |c_\alpha|^2(1-\frac{t}{2N}\Delta_\beta^2)^{2N}}\nonumber\\
&&=\frac{|c_\beta|^2}{|c_\beta|^2+(1-|c_\beta|^2)\left(1-\frac{1}{\frac{2N}{t\Delta_\beta^2}}\right)^{\frac{2N}{t\Delta_\beta^2}t\Delta_\beta^2}}\nonumber\\
&&\geq\frac{|c_\beta|^2}{|c_\beta|^2+(1-|c_\beta|^2)e^{-t\Delta_\beta^2}},
\end{eqnarray}
where in the second line, we use the Taylor approximation for $\cos\left(\sqrt{\frac{t}{N}}h_\alpha\right)$ since $\sqrt{\frac{t}{N}}\ll1$, and in the last line, we use the fact that $(1-1/x)^x< e^{-1}$. As a result, to make the overlap $p_{0,\beta}$ be at least $1-\zeta$, we therefore require the overall complexity
\begin{eqnarray}\label{eplin}
\text{Lindlad.evo.time (t)}=\mathcal{O}\left(|c_\beta|^{-1}\Delta_\beta^{-2}\log(\zeta^{-1})\right),
\end{eqnarray}
which corresponds to the dilated Hamiltonian simulation time
\begin{eqnarray}
\text{Hami.simu.time (T)}=\mathcal{O}\left(|c_\beta|^{-1}\Delta_\beta^{-4}\log^2(\zeta^{-1})\varepsilon^{-1}\right),
\end{eqnarray}
where the $|c_\beta|^{-1}$ comes from using amplitude amplification to amplify $p_0$ and we ignore the additional logarithmic factors. Note that $\varepsilon$ dependence just affects the Lindbladian simulation error but will not contribute to the accuracy of the eigenstate preparation.

\section{Lindbladian fast forwarding\label{smff}}

\begin{figure}[htbp]
\centering
\includegraphics[width=0.9\textwidth]{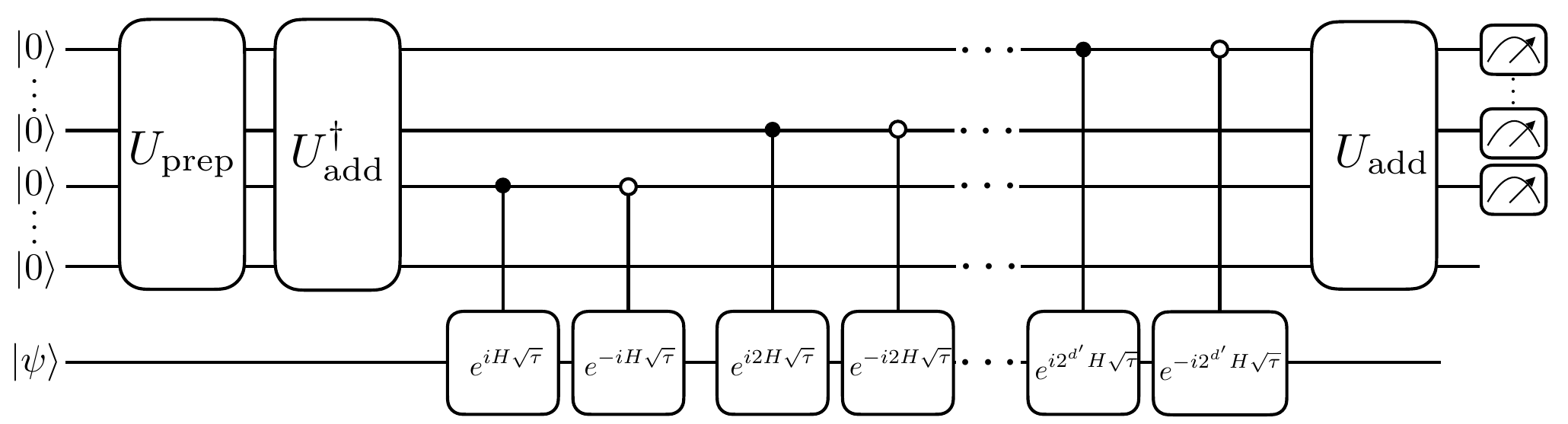}
\caption{The quantum circuit of our Lindbladian fast-forwarding algorithm for Eq.~(\ref{mainslme}) in Theorem \ref{the2}. The fast-forwarding algorithm can also be used for QPE tasks with a Heisenberg-limit scaling (Theorem \ref{the3}).\label{fig.fflind}}
\end{figure}

Here, we give a quadratically fast forwarding Lindbladian simulation algorithm for Eq. \ref{slme}. Recall that Eq. \ref{shorth} can be re-expressed in ancilla $+/-$ basis
\begin{eqnarray}
e^{-i \tilde{H}\sqrt{\tau}}\frac{1}{\sqrt{2}}\left(|+\rangle+|-\rangle\right)|\psi\rangle=\frac{1}{\sqrt{2}}|+\rangle e^{-iH\sqrt{\tau}}|\psi\rangle+\frac{1}{\sqrt{2}}|-\rangle e^{iH\sqrt{\tau}}|\psi\rangle,
\end{eqnarray}
putting which in Eq. \ref{lint2} gives
\begin{eqnarray}\label{lint3}
(\textbf{H}^{\otimes N}\otimes I)\left(\prod_i e^{-i\tilde{H}_i\sqrt{\frac{t}{N}}}\right)|0\rangle^{\otimes N}|\psi\rangle&&=\sum_{m=0}^N\frac{\sqrt{\binom{N}{m}}}{2^{N/2}}|\bar{m}\rangle e^{-iH\sqrt{\frac{t}{N}}(2m-N)}|\psi\rangle\nonumber\\
&&=\sum_{m=0}^N\frac{\sqrt{\binom{N}{m}}}{2^{N/2}}|\bar{m}\rangle |\psi(2m-N)\rangle,
\end{eqnarray}
where we set $|\psi(2m-N)\rangle=e^{-iH\sqrt{\frac{t}{N}}(2m-N)}|\psi\rangle$. A naive way of preparing Eq. \ref{lint3} is to first prepare the state
\begin{eqnarray}
\sum_{m=0}^N\frac{\sqrt{\binom{N}{m}}}{2^{N/2}}|\bar{m}\rangle |\psi\rangle,
\end{eqnarray}
and then implement the controlled unitary
\begin{eqnarray}\label{wholl}
\sum_{m=0}^N |\bar{m}\rangle\langle \bar{m}|\otimes e^{-iH\sqrt{\frac{t}{N}}(2m-N)},
\end{eqnarray}
which, however, will require a Hamiltonian simulation time of order $\mathcal{O}(\sqrt{Nt})=\mathcal{O}(t^2/\varepsilon)$ according to Eq. \ref{dilated} and $N$ ancilla qubits.

To enable Lindbladian fast forwarding and save ancilla qubits, we can observe that in Eq. \ref{lint3}, the contribution (probality) of $|\psi(2m-N)\rangle$ in the reduced system density matrix is $p_m=2^{-N}\binom{N}{m}$, which is a binomial distribution concentrated on $m=N/2$, corresponding to $|\psi(0)\rangle=|\psi\rangle$. According to Lemma \ref{con2}, we have
\begin{eqnarray}
\bm{\Pr}(|m-N/2|> cN)=\sum_{|m-N/2|\geq cN}2^{-N}\binom{N}{m}<2\exp(-2 c^2N).
\end{eqnarray}
Therefore, instead of the above naive way of simulation, we can only prepare the components around $m=N/2$: $m\in[N/2-cN,N/2+cN-1]$.
Without loss of generality, we assume $N$ to be even.
Also, since the Lindbladian simulation only cares about the system reduced density matrix, we can replace the $N$-qubit Dicke state $|\bar{m}\rangle$ by a $\log(N)$-qubit computational basis state $|m\rangle$ to exponentially save the number of ancilla qubits. As a result, our goal can be simplified to prepare the following state
\begin{eqnarray}\label{goal}
|\text{Goal}\rangle=\sum_{m=N/2-cN}^{N/2+cN-1}\frac{\sqrt{\binom{N}{m}}}{2^{N/2}}|m\rangle |\psi(2m-N)\rangle+\sum_{m<N/2-cN;m>N/2+cN-1}\frac{\sqrt{\binom{N}{m}}}{2^{N/2}}|m\rangle|\Psi_m\rangle,
\end{eqnarray}
where $|\Psi_m\rangle$ can be arbitrary quantum states. The difference between the resulting reduced system density matrix of $|\text{Goal}\rangle$ and the exact one can be upper-bounded
\begin{eqnarray}\label{tracebound}
&&\left\|\sum_{m=0}^N2^{-N}\binom{N}{m} |\psi(2m-N)\rangle\langle\psi(2m-N)|-\Tr_a(|\text{Goal}\rangle\langle \text{Goal}|)\right\|_1\nonumber\\
=&&\left\|\sum_{m<N/2-cN;m>N/2+cN-1}\frac{\binom{N}{m}}{2^{N}}(|\Psi_m\rangle\langle\Psi_m|-|\psi(2m-N)\rangle\langle\psi(2m-N|))\right\|_1\nonumber\\
\leq&& \sum_{m<N/2-cN;m>N/2+cN-1}\frac{\binom{N}{m}}{2^{N}}\left\| |\Psi_m\rangle\langle\Psi_m|-|\psi(2m-N)\rangle\langle\psi(2m-N|\right\|_1\nonumber\\
\leq&& \sum_{m<N/2-cN;m>N/2+cN-1}\frac{\binom{N}{m}}{2^{N}}<2\exp(-2 c^2N)
\end{eqnarray}
where in the last line, the first equality is because the operators in $\|\cdot\|_1$ of the second-to-last line are normalized density matrices. To make the error at most $\varepsilon'$, we require
\begin{eqnarray}\label{derep}
&&2\exp(-2c^2N)= \varepsilon'
\nonumber\\\rightarrow&&c=\sqrt{\frac{\log(2/\varepsilon')}{2}}\frac{1}{\sqrt{N}}\nonumber\\\rightarrow&& cN=\mathcal{O}\left(N^{1/2}\log^{1/2}(\varepsilon'^{-1})\right).
\end{eqnarray}
Let $d=\lceil \log N\rceil$ and $d'=\lceil \log(cN)\rceil+1$, to prepare Eq. \ref{goal}, we can first prepare the state 
\begin{eqnarray}\label{iinn}
\sum_{m=0}^{N}\frac{\sqrt{\binom{N}{m}}}{2^{N/2}}|m\rangle |\psi\rangle,
\end{eqnarray}
where $|m\rangle$ is a $d$-qubit computational basis state under binary representation. The detailed analysis of how to prepare $\sum_{m=0}^{N}\frac{\sqrt{\binom{N}{m}}}{2^{N/2}}|m\rangle$ can be found in Section \ref{sec.stateprep}.

Next, we can introduce the in-place addition~\cite{cuccaro2004new,draper2004logarithmic} (modular arithmetic) unitary $U_{\text{add}}$ with the property
\begin{eqnarray}
U_{\text{add}}|m\rangle=|m+N/2-2^{d'-1} \mod N\rangle,
\end{eqnarray}
for $m\in [N]$.
With this, we construct the following unitary
\begin{eqnarray}\label{vhh}
V_H=&&(U_{\text{add}}\otimes I)\left(I\otimes e^{iH\sqrt{\frac{t}{N}}}\right)\prod_{j=0}^{d'-1}\left(|0\rangle\langle 0|_j\otimes e^{iH\sqrt{\frac{t}{N}}2^{j}}+|1\rangle\langle 1|_j\otimes e^{-iH\sqrt{\frac{t}{N}}2^{j}}\right)(U_{\text{add}}^\dag\otimes I)\nonumber\\
=&&(U_{\text{add}}\otimes I)\sum_{m=0}^{2^{d'}-1} |m\rangle\langle m|\otimes e^{-iH\sqrt{\frac{t}{N}}(2m-2^{d'})}(U_{\text{add}}^\dag\otimes I)\nonumber\\&&+(U_{\text{add}}\otimes I)\sum_{m=2^{d'}}^{2^{d}-1} |m\rangle\langle m|\otimes e^{-iH\sqrt{\frac{t}{N}}(2(m\mod 2^{d'})-2^{d'})}(U_{\text{add}}^\dag\otimes I) \nonumber\\=&&\sum_{m=N/2-2^{d'-1}}^{N/2+2^{d'-1}-1} |m\rangle\langle m|\otimes e^{-iH\sqrt{\frac{t}{N}}(2m-N)}\nonumber\\&&+\sum_{m<N/2-2^{d'-1};m>N/2+2^{d'-1}-1}|m\rangle\langle m|\otimes e^{-iH\sqrt{\frac{t}{N}}(2(m-N/2+2^{d'-1}\mod 2^{d'})-2^{d'})}
\end{eqnarray}
where we only use the first $d'$ ancilla qubits for controlled Hamiltonian simulation and leave the others as identities.
The circuit of $V_H$ can be found in Fig. \ref{fig.fflind}.
The first term of $V_H$ in the last line of Eq. \ref{vhh} plays the dominant role.
Implementing $V_H$ on the initial state Eq. \ref{iinn} will exactly give us the form of $|\text{Goal}\rangle$, and we realize the simulation of the Lindbladian Eq. \ref{slme}.

The construction cost of $V_H$ can be characterized by the required Hamiltonian simulation time $\sqrt{\frac{t}{N}}(1+2+\cdots+2^{d'-1})=\mathcal{O}(c\sqrt{Nt})=\mathcal{O}(t^{1/2}\log^{1/2}(\varepsilon'^{-1}))$. Note that there is no dependence on $N$, we can set $N$ to be arbitrarily large without changing the required Hamiltonian simulation time. Since when $N=t^3/\varepsilon^{-2}$, the system reduced density matrix of Eq. \ref{lint3} will have at most $\varepsilon$-error in trace distance from the density matrix evolved from the exact Lindbladian, when setting $\varepsilon'=\varepsilon$ and $N=\mathcal{O}(t^3\varepsilon^{-2})$, the error between our fast-forwarding algorithm and the exact Lindbladian is also $\mathcal{O}(\varepsilon)$. Thus, to simulate the Lindbladian Eq. \ref{slme} within an error $\varepsilon$, our algorithm has the complexity
\begin{eqnarray}\label{ff}
\text{Hami.simu.time }&&=c\sqrt{Nt}=\mathcal{O}\left(t^{1/2}\log^{1/2}(\varepsilon^{-1})\right),\\
\text{Number of ancilla qubits}&&=\log(N)=\mathcal{O}\left(\log(t)+\log(\varepsilon^{-1})\right).
\end{eqnarray}

\section{Lindbladian fast forwarding for fast QPE\label{smffqpe}}

In using Lindbladian for QPE, we measure the ancilla 0/1 basis (Eq. \ref{lint2}), while in our fast-forwarding algorithm, we adopt the ancilla +/- basis (Eq. \ref{lint3}). The two expressions have the connection
\begin{eqnarray}\label{unary}
&&(\textbf{H}^{\otimes N}\otimes I)\sum_{m=0}^N\frac{\sqrt{\binom{N}{m}}}{2^{N/2}}|\bar{m}\rangle |\psi(2m-N)\rangle=\nonumber\\
&&\sum_{m=0}^N (-i)^m\sqrt{\binom{N}{m}}|\bar{m}\rangle\cos\left(\sqrt{\frac{t}{N}}H\right)^{N-m}\sin\left(\sqrt{\frac{t}{N}}H\right)^{m}|\psi\rangle.
\end{eqnarray}
If we replace the $N$-qubit Dicke state $|\bar{m}\rangle$ by a $\log(N)$-qubit computational basis state $|m\rangle$, we can expect that there exists a unitary operator $U_D$ such that
\begin{eqnarray}\label{bbinary}
&&(U_D\otimes I)\underbrace{\sum_{m=0}^N\frac{\sqrt{\binom{N}{m}}}{2^{N/2}}|m\rangle |\psi(2m-N)\rangle}_{|S_{+/-}\rangle}=\nonumber\\
&&\underbrace{\sum_{m=0}^N (-i)^m\sqrt{\binom{N}{m}}|m\rangle\cos\left(\sqrt{\frac{t}{N}}H\right)^{N-m}\sin\left(\sqrt{\frac{t}{N}}H\right)^{m}|\psi\rangle}_{|S_{0/1}\rangle}.
\end{eqnarray}
To achieve this, according to the results in Ref. \cite{bernard2024dynamical}, $U_D$ should satisfy
\begin{eqnarray}
\langle m|U_D|m'\rangle&&=\langle \bar{m}|\textbf{H}^{\otimes N}|\bar{m}'\rangle\nonumber\\
&&=\sqrt{\binom{N}{m}\binom{N}{m'}}\mathcal{K}_{m'}(m,0.5,N),
\end{eqnarray}
where $\mathcal{K}_{m'}(m,0.5,N)$ is the Kravchuk polynomial~\cite{krawtchouk1929generalisation}. Our Lindbladian fast-forwarding algorithm (Eq. \ref{ff}) with a Hamiltonian simulation time $\mathcal{O}\left(t^{1/2}\log^{1/2}(\varepsilon^{-1})\right)$ can prepare $|\text{Goal}\rangle$ (Eq. \ref{goal}), which following the similar derivation in Eq. \ref{tracebound}, can be shown to satisfy
\begin{eqnarray}\label{vnorm}
&&\|(U_D\otimes I)|\text{Goal}\rangle-|S_{0/1}\rangle\|_2=\||\text{Goal}\rangle-|S_{+/-}\rangle\|_2\nonumber\\&&=\left\|\sum_{m<N/2-cN;m>N/2+cN-1}2^{-N/2}\sqrt{\binom{N}{m}}|m\rangle |\Psi_m\rangle\right\|_2\nonumber\\&&
\leq \sqrt{\sum_{m<N/2-cN;m>N/2+cN-1}2^{-N}\binom{N}{m}}=\varepsilon^{1/2},
\end{eqnarray}
where the last line is following Eq. \ref{derep} and we set $\varepsilon'=\varepsilon$ as explained in the text above Eq. \ref{ff}.

\subsection{Eigenvalue estimation}

We set $|\psi\rangle=|\psi_\alpha\rangle$ and write 
\begin{eqnarray}
|S_{0/1}\rangle&&=\sum_{m=0}^N \sqrt{p_m}|m\rangle |\phi_m\rangle\nonumber\\
(U_D\otimes I)|\text{Goal}\rangle&&=\sum_{m=0}^N \sqrt{p'_m}|m\rangle |\phi'_m\rangle.
\end{eqnarray}
Recall Eq. \ref{eslin}, to have an $\epsilon$-estimation with success probability at least $1-\delta$, we require the Lindbladian simulation time $t$ in $|S_{0/1}\rangle$ to be $t=\mathcal{O}(\epsilon^{-2}\log (\delta^{-1}))$. The way we obtain this result is through Eq. \ref{c1}, where we show that
\begin{eqnarray}
\sum_{|m-Nq_\alpha|\geq 2th_\alpha\epsilon}p_m\leq 2\exp\left(-\frac{24}{11}t\epsilon^2\right).
\end{eqnarray}
Now, given the state $(U_D\otimes I)|\text{Goal}\rangle$, our way to estimate the eigenvalue is also by measuring the ancilla qubits, and using the measurement result $m$ and the relation $m/N=\sin(\sqrt{\frac{t}{N}}h_\alpha)^2$ to infer $h_\alpha$. According to Eq. \ref{vnorm}, we have
\begin{eqnarray}\label{xyg}
\varepsilon=\|(U_D\otimes I)|\text{Goal}\rangle-|S_{0/1}\rangle\|_2^2&&=\sum_{m=0}^N\|\sqrt{p'_m}|m\rangle |\phi'_m\rangle-\sqrt{p_m}|m\rangle |\phi_m\rangle\|_2^2\nonumber\\
&&\geq \sum_{m=0}^N \left|\sqrt{p'_m}-\sqrt{p_m}\right|^2\geq \frac{\left(\sum_{m=0}^N|p_m-p'_m|\right)^2}{\sum_{m'=0}^N \left|\sqrt{p'_{m'}}+\sqrt{p_{m'}}\right|^2} \nonumber\\
&&=\frac{\left(\sum_{m=0}^N|p_m-p'_m|\right)^2}{2+\sum_{m'=0}^N 2\sqrt{p'_{m'} p_{m'}}} \geq \frac{\left(\sum_{m=0}^N|p_m-p'_m|\right)^2}{4},
\end{eqnarray}
where the second inequality in the second line is due to the Cauchy inequality. Eq.~\ref{xyg} implies
\begin{eqnarray}
\sum_{m=0}^N|p_m-p'_m|\leq2\varepsilon^{1/2}.
\end{eqnarray}
Therefore, we have bounds on $p'_m$
\begin{eqnarray}
\sum_{|m-Nq_\alpha|\geq 2th_\alpha\epsilon}p'_m&&=\sum_{|m-Nq_\alpha|\geq 2th_\alpha\epsilon}(p'_m-p_m+p_m)\nonumber\\
&&\leq \sum_{|m-Nq_\alpha|\geq 2th_\alpha\epsilon}(|p'_m-p_m|+p_m)\nonumber\\&&
=\mathcal{O}\left(\exp(-t\epsilon^2)+\varepsilon^{1/2}\right).
\end{eqnarray}
Thus, to have an estimation error within $\epsilon$ and with a success probability at least $1-\delta$, it is sufficient to set $\varepsilon^{1/2}=\mathcal{O}(\delta)$ and $t=\mathcal{O}(\epsilon^{-2}\log(\delta^{-1}))$. Plugging this to Eq. \ref{ff} (and ignoring the logarithmic factors), we can give the complexity of using Lindbladian fast-forwarding algorithm for eigenvalue estimation
\begin{eqnarray}
\text{Hami.simu.time }&&=\mathcal{O}(\log^{1/2}(\varepsilon^{-1})t^{1/2})\nonumber\\
&&=\mathcal{O}(\epsilon^{-1}\log(\delta^{-1})).
\end{eqnarray}

\subsection{Eigenstate preparation}
We set $|\psi\rangle=\sum_\alpha \Pi_\alpha |\psi\rangle=\sum_\alpha c_\alpha |\psi_\alpha\rangle$ and we want to prepare $|\psi_\beta\rangle$ given the knowledge of $h_\beta$. Again, we can let $h_\beta$ be 0 by shifting the spectrum of $H$. Recall that in Section \ref{slowep}, to obtain the complexity Eq. \ref{eplin}, the core is the bound of $p_0$ in Eq. \ref{import1} and the bound of $p_{0,\beta}$ in Eq. \ref{import2}. For $(U_D\otimes I)|\text{Goal}\rangle$, this corresponds to 
\begin{eqnarray}
p_0'&&=\|(\langle0|\otimes I)(U_D\otimes I)|\text{Goal}\rangle\|_2^2,\\
p_{0,\beta}'&&=\frac{|\langle0| \langle\psi_\beta|(U_D\otimes I)|\text{Goal}\rangle|^2}{p_0'}.
\end{eqnarray}
Due to Eq. \ref{vnorm}, we have
\begin{eqnarray}
\sqrt{p_0'}&&=\|(\langle0|\otimes I)(U_D\otimes I)|\text{Goal}\rangle\|_2\nonumber\\
&&=\left\|(\langle0|\otimes I)\left((U_D\otimes I)|\text{Goal}\rangle-|S_{0/1}\rangle+|S_{0/1}\rangle\right)\right\|_2\nonumber\\
&&\geq \left|\|(\langle0|\otimes I)|S_{0/1}\rangle\|_2-\left\|(\langle0|\otimes I)\left((U_D\otimes I)|\text{Goal}\rangle-|S_{0/1}\rangle\right)\right\|_2\right|\nonumber\\
&&\geq\sqrt{p_0}-\varepsilon^{1/2}\geq |c_\beta|-\varepsilon^{1/2},\nonumber\\
\end{eqnarray}
where in the derivations, we set $\varepsilon^{1/2}$ to be smaller than $\sqrt{p_0}$ and we utilze Eq. \ref{import1}.
Similarly, it is also easy to show $\sqrt{p_0'}\leq\sqrt{p_0}+\varepsilon^{1/2}$. For $p_{0,\beta}'$, we have
\begin{eqnarray}
\sqrt{p_{0,\beta}'}&&=\frac{|\langle0| \langle\psi_\beta|(U_D\otimes I)|\text{Goal}\rangle|}{\sqrt{p_0'}}\nonumber\\
&&=\frac{\left|\langle0| \langle\psi_\beta|\left((U_D\otimes I)|\text{Goal}\rangle-|S_{0/1}\rangle+|S_{0/1}\rangle\right)\right|}{\sqrt{p_0'}}\nonumber\\&&\geq \frac{\left|\left|\langle0| \langle\psi_\beta|\left((U_D\otimes I)|\text{Goal}\rangle-|S_{0/1}\rangle\right)\right|-\left|\langle0| \langle\psi_\beta|S_{0/1}\rangle\right|\right|}{\sqrt{p_0'}}\nonumber\\
&&\geq\frac{\left|\langle0| \langle\psi_\beta|S_{0/1}\rangle\right|-\varepsilon^{1/2}}{\sqrt{p_0'}}\nonumber\\
&&\geq\frac{\left|\langle0| \langle\psi_\beta|S_{0/1}\rangle\right|-\varepsilon^{1/2}}{\sqrt{p_0}+\varepsilon^{1/2}}\nonumber\\
&&=\frac{\sqrt{p_{0,\beta}}-\varepsilon^{1/2}/\sqrt{p_0}}{1+\varepsilon^{1/2}/\sqrt{p_0}}\nonumber\\&&\geq\frac{\sqrt{p_{0,\beta}}-\varepsilon^{1/2}/|c_\beta|}{1+\varepsilon^{1/2}/|c_\beta|},
\end{eqnarray}
where in the derivations, we set $\varepsilon^{1/2}$ to be also smaller than $\left|\langle0| \langle\psi_\beta|S_{0/1}\rangle\right|$, and we utilize Eq. \ref{import2}. Now, we can set $\varepsilon^{1/2}=|c_\beta|\zeta$ and $p_{0,\beta}=1-\zeta$ which require the Lindbladian evolution time $t=\mathcal{O}\left(\Delta_\beta^{-2}\log(\zeta^{-1})\right)$.
Then we will have
\begin{eqnarray}
\sqrt{p_0'}\geq |c_\beta|-\varepsilon^{1/2}=|c_\beta|(1-\zeta),
\end{eqnarray}
and 
\begin{eqnarray}
\sqrt{p_{0,\beta}'}&&\geq\frac{\sqrt{p_{0,\beta}}-\varepsilon^{1/2}/|c_\beta|}{1+\varepsilon^{1/2}/|c_\beta|}=\frac{\sqrt{1-\zeta}-\zeta}{1+\zeta}\nonumber\\
&&\geq \frac{1-2\zeta}{1+\zeta}=1-\frac{3\zeta}{1+\zeta}\geq 1-3\zeta\nonumber\\
&&\geq \sqrt{1-6\zeta}.
\end{eqnarray}
In other words, we will have a preparation inaccuracy $1-p_{0,\beta}'=\mathcal{O}(\zeta)$, which fullfill our requirement. Putting these bounds to Eq. \ref{ff} and combining the amplitude amplification algorithm, the overall resulting complexity is 
\begin{eqnarray}
\text{Hami.simu.time }&&=\mathcal{O}(\log^{1/2}(\varepsilon^{-1})t^{1/2})\nonumber\\
&&=\mathcal{O}(|c_\beta|^{-1}\Delta_\beta^{-1}\log(\zeta^{-1})).
\end{eqnarray}

At the end of this section, we want to mention that, here, we directly assume we are given access to $U_D$ since we currently do not know how to construct $U_D$ through elementary quantum gates, which is an open question in this work. Nevertheless, one can always go back to $|\bar{m}\rangle$ Dicke state encoding and simply use $\textbf{H}^{\otimes N}$ to avoid the introduction of $U_D$, i.e. go back from Eq. \ref{bbinary} to Eq. \ref{unary}. While this will increase the number of ancilla qubits from $\log(N)$ to $N$, the Heisenberg limit scaling of the Hamiltonian simulation time in the QPE tasks will remain unchanged.

\section{QPE, amplitude estimation, and quadratic quantum speedup\label{smqua}}
We consider a decision problem where we judge whether there exists $x\in\{0,1\}^n$ such that $f(x)=1$. The corresponding quantum oracle $U_f$ has the action 
\begin{eqnarray}
U_f|x\rangle|0\rangle=|x\rangle|f(x)\rangle.
\end{eqnarray}
Suppose the number of $x$ such that $f(x)=1$ is $W$, then we can find that
\begin{eqnarray}
A=\left\|(I_n\otimes \langle 1|)U_f(\textbf{H}^{\otimes n}\otimes I)|0\rangle^{\otimes n}|0\rangle\right\|_2=\left\|2^{-n/2}(I_n\otimes \langle 1|)\sum_x|x\rangle|f(x)\rangle\right\|_2=2^{-n/2}\sqrt{W}.
\end{eqnarray}
Therefore, judging whether $A$ is zero or not can solve the decision problem. The amplitude estimation algorithm \cite{brassard2000quantum} is as follows. We first construct the following unitary
\begin{eqnarray}
U_{ae}=(I_{n+1}-2U_f(\textbf{H}^{\otimes n}\otimes I)|0\rangle^{\otimes n}|0\rangle\langle 0|^{\otimes n}\langle 0|(\textbf{H}^{\otimes n}\otimes I)U_f^\dag)(I_n\otimes Z),
\end{eqnarray}
and let $U_{ae}$ replace $e^{-iH}$ in QPE (either the standard QPE (Eq. \ref{stanhs}) or our fast-forwarding Lindbladian-based QPE (Eq. \ref{vhh})). Next, we set the input $|\psi\rangle=\sum_x|x\rangle|f(x)\rangle$. Then, running the QPE algorithms with queries on $U_{ae}$ for at most $\mathcal{O}(2^{n/2})$ times is sufficient to judge whether $W$ is zero or not. The $\mathcal{O}(2^{n/2})$ scaling is quadratically better than classical algorithms $\mathcal{O}(2^{n})$, and the reason for this quadratic speedup is exactly the Heisenberg limit scaling of QPE. Note that, for our Lindbladian-based QPE, we are doing the fast-forwarding simulation for the Lindbladian Eq. \ref{slme} with $H_{ae}$ to be the jump operator where $H_{ae}$ satisfies $e^{-iH_{ae}}=U_{ae}$.

\section{Quantum Gibbs state preparation\label{smpar}}
Here, we describe how to use our fast-forwarding result for Gibbs state preparations. The basic idea is to combine our Lindbladian fast-forwarding simulation algorithm with the Lindbladian-based differential equation solver introduced in \cite{shang2024design}. Given a positive semidefinite Hamiltonian $H_P$ with $\|H_P\|\leq 1$, which is also the hidden assumption in previous algorithms \cite{gilyen2019quantum,chowdhury2016quantum}, the Gibbs state $\rho_\beta$ and the partition function $Z_\beta$ at the inverse temperature $\beta$ are
\begin{eqnarray}
\rho_\beta=\frac{e^{-\beta H_P}}{Z_\beta},~\quad Z_\beta=\Tr(e^{-\beta H_P}).
\end{eqnarray}
For $\rho_\beta$, we consider its symmetric purification state
\begin{eqnarray}
|\rho_\beta\rangle&&=\sqrt{\frac{2^n}{Z_\beta}}(e^{-\beta H_P/2}\otimes I)|\Omega\rangle\nonumber\\
&&=\sqrt{\frac{2^n}{Z_\beta}}e^{-\frac{\beta}{2} H_P\otimes I}|\Omega\rangle,
\end{eqnarray}
where $|\Omega\rangle=2^{-n/2}\sum_i |i\rangle|i\rangle$. 

Based on Ref.~\cite{shang2024design}, if we start from an initial state $|+\rangle\langle+|\otimes |\Omega\rangle\langle\Omega|$, and construct a Lindbladian of the form Eq. \ref{flme}, with the single jump operator having the form
\begin{eqnarray}\label{jumpg}
F=\begin{pmatrix}
\sqrt{H_P}\otimes I &0\\
0 & 0\end{pmatrix},
\end{eqnarray}
then at the time $\beta/2$, the output density matrix will have the form
\begin{eqnarray}
\rho_{out}=\frac{1}{2}\begin{pmatrix}
\cdot &e^{-\frac{\beta}{2} H_P\otimes I}|\Omega\rangle\langle\Omega|\\|\Omega\rangle\langle\Omega|e^{-\frac{\beta}{2} H_P\otimes I}
 & \cdot \end{pmatrix}=\frac{1}{2}\begin{pmatrix}
\cdot & \sqrt{\frac{Z_\beta}{2^n}}|\rho_\beta\rangle\langle\Omega|\\\sqrt{\frac{Z_\beta}{2^n}}|\Omega\rangle\langle\rho_\beta|
 & \cdot \end{pmatrix}.
\end{eqnarray}
Note that to construct $F$, we assume that we are given the access (or the block encoding) of $\sqrt{H_P}$. If we have access to the preparation unitary of the purification of $\rho_{out}$, which is true in our algorithm, then we can follow the amplitude amplification procedure introduced in Ref.~\cite{shang2024design} to finally prepare $|\rho_\beta\rangle$ with $\mathcal{O}(\sqrt{\frac{2^n}{Z_\beta}})$ queries to the preparation unitary. 

Let us now consider the complexity. We can observe that the jump operator Eq. \ref{jumpg} is Hermitian, which meets our fast-forwarding requirements. Combining the fast-forwarding result Eq. \ref{ff} and the amplitude amplification, the total required Hamiltonian simulation time ($e^{-i F t}$) to prepare $|\rho_\beta\rangle$ within an error $\varepsilon$ is
\begin{eqnarray}\label{parsi}
\text{Hami.simu.time }&&=\tilde{\mathcal{O}}\left(\sqrt{\frac{2^n}{Z_\beta}}\beta^{1/2}\log^{1/2}(\varepsilon^{-1})\right),
\end{eqnarray}
where the $\tilde{\mathcal{O}}$ symbol hides the additional $\log(\sqrt{\frac{2^n}{Z_\beta}})$ factors.

To make a fair comparison with previous methods, we also need to consider the cost of implementing Hamiltonian simulation on quantum computers. To simulate a Hamiltonian $e^{-iF t}$ of time $t$ within an error $\varepsilon$, the state-of-the-art quantum singal processing technique needs to query the block encoding of $\sqrt{H_P}$ for $\tilde{\mathcal{O}}(t+\log(\varepsilon^{-1}))$ times \cite{low2017optimal}. Recall Fig. \ref{fig.fflind}, the whole controlled Hamiltonian simulation Eq. \ref{wholl} is decomposed into $d'=\tilde{\mathcal{O}}(\log(t^{1.5}\varepsilon^{-1}))$ single-control Hamiltonian simulations. To make sure the whole error is within $\varepsilon$, we can simply let each single-control Hamiltonian simulation be within an error $\varepsilon/d'$. The resulting overall query complexity on $\sqrt{H_P}$ to construct Eq. \ref{wholl} within $\varepsilon$ then becomes
\begin{eqnarray}
\tilde{\mathcal{O}}(t+d'\log(\varepsilon^{-1}))=\tilde{\mathcal{O}}(t+\log^2(\varepsilon^{-1})).
\end{eqnarray}
Putting this result in Eq. \ref{parsi}, the final query complexity on $\sqrt{H_P}$ to prepare $|\rho_\beta\rangle$ (and therefore $\rho_\beta$) within an error $\varepsilon$ becomes
\begin{eqnarray}
&&\tilde{\mathcal{O}}\left(\sqrt{\frac{2^n}{Z_\beta}}\beta^{1/2}\log^{1/2}(\varepsilon^{-1})+\log^2(\varepsilon^{-1})\right)\nonumber\\
=&&\tilde{\mathcal{O}}\left(\sqrt{\frac{2^n}{Z_\beta}}\sqrt{\max \left(\beta, \log^3(\varepsilon^{-1})\right)\log(\varepsilon^{-1})} \right).
\end{eqnarray}
When $\beta>\log^3(\varepsilon^{-1})$, the complexity becomes $\tilde{\mathcal{O}}\left(\sqrt{\frac{2^n}{Z_\beta}}\sqrt{\beta\log(\varepsilon^{-1})} \right)$.

In comparison, given the $\sqrt{H_P}$ access, using the previous state-of-the-art QSVT \cite{gilyen2019quantum} technique, the query complexity is 
\begin{eqnarray}
\tilde{\mathcal{O}}\left(\sqrt{\frac{2^n}{Z_\beta}}\sqrt{\max \left(\beta, \log(\varepsilon^{-1})\right)\log(\varepsilon^{-1})}\right),
\end{eqnarray}
when $\beta>\log(\varepsilon^{-1})$, the complexity becomes $\tilde{\mathcal{O}}\left(\sqrt{\frac{2^n}{Z_\beta}}\sqrt{\beta\log(\varepsilon^{-1})}\right)$, which is also the result in another Gibbs state preparation algorithm \cite{chowdhury2016quantum} based on the $\sqrt{H_P}$ access. Thus, for large $\beta$, our algorithm matches the scaling on the state-of-the-art dependence of $\varepsilon$ in QSVT.

We want to mention that the above procedure can also be changed into a partition function estimation algorithm \cite{bravyi2021complexity} by simply changing the amplitude amplification to amplitude estimation \cite{brassard2000quantum}.

\section{Generalizations of Lindbladian fast forwarding\label{smgen}}
Here, we give a much broader class of Lindbladians with Eq. \ref{slme} as a subset that can enable quadratic fast forwardings. We refer to this class of Lindbladians as the Choi commuting Lindbladians. The form of Choi commuting Lindbladians is
\begin{equation}\label{clme}
\frac{\d\rho(t)}{\d t}=\mathcal{L}_C[\rho(t)]=\sum_{i=1}^K\left(H_i\rho(t) H_i^\dag-
\frac{1}{2}\{\rho(t),H_i^\dag H_i\}\right),
\end{equation}
where we require $H_i$ to be Hermitian and
\begin{equation}\label{choi}
[H_i\otimes H_i^*-\frac{1}{2}H_i^2\otimes I-I\otimes \frac{1}{2}H_i^{*2},H_j\otimes H_j^*-\frac{1}{2}H_j^2\otimes I-I\otimes \frac{1}{2}H_j^{*2}]=0,~\forall\,i,j\in[1,K].
\end{equation}
To see why Eq. \ref{clme} enables fast forwarding, we need to introduce the vectorization picture where we map $\rho(t)=\sum_{ij}\rho_{ij}(t)|i\rangle\langle j|$ to its Choi vector $\ket{\rho(t)\rangle}=\sum_{ij}\rho_{ij}(t)|i\rangle|j\rangle$. Under this vectorization picture, the Lindbladian Eq. \ref{clme} has the form
\begin{eqnarray}
\frac{\d\ket{\rho(t)\rangle}}{\d t}&&=L_C\ket{\rho(t)\rangle}=\sum_{i=1}^K\left(H_i\otimes H_i^*-\frac{1}{2}H_i^\dag H_i\otimes I-I\otimes \frac{1}{2}H_i^TH_i^*\right)\ket{\rho(t)\rangle}\nonumber\\
&&=\sum_{i=1}^K\left(H_i\otimes H_i^*-\frac{1}{2}H_i^2\otimes I-I\otimes \frac{1}{2}H_i^{*2}\right)\ket{\rho(t)\rangle}.
\end{eqnarray}
which has the solution
\begin{eqnarray}\label{vector}
\ket{\rho(t)\rangle}&&=\exp\left(\sum_{i=1}^K\left(H_i\otimes H_i^*-\frac{1}{2}H_i^2\otimes I-I\otimes \frac{1}{2}H_i^{*2}\right)t\right)\ket{\rho(0)\rangle}\nonumber\\
&&=\prod_{i=1}^K \exp\left(\left(H_i\otimes H_i^*-\frac{1}{2}H_i^2\otimes I-I\otimes \frac{1}{2}H_i^{*2}\right)t\right)\ket{\rho(0)\rangle},
\end{eqnarray}
where the equality in the second line is due to Eq. \ref{choi}. Eq. \ref{vector} indicates that $\rho(t)$ can be prepared by the following channel
\begin{eqnarray}
\rho(t)=e^{\mathcal{L}_{H_K}t}\left[e^{\mathcal{L}_{H_{K-1}}t}\left[\cdots e^{\mathcal{L}_{H_1}t}\left[\rho(0)\right]\right]\right],
\end{eqnarray}
with each $e^{\mathcal{L}_{H_i}t}[\cdot]$ corresponds to the Lindbladian dynamics for a time $t$
\begin{eqnarray}
\frac{\d\rho(t)}{\d t}=H_i\rho(t) H_i-
\frac{1}{2}\{\rho(t),H_i^2\}.
\end{eqnarray}
We can find that each Lindbladian $e^{\mathcal{L}_{H_i}t}[\cdot]$ fits the form of Eq. \ref{slme} and therefore, its simulation enables a quadratic Lindbladian fast forwarding with the complexity in Eq. \ref{ff}. As a result, the overall dynamics Eq. \ref{clme} can be simulated with the complexity
\begin{eqnarray}
\text{Hami.simu.time}&&=\mathcal{O}\left(Kt^{1/2}\log^{1/2}(\varepsilon^{-1})\right),\nonumber\\
\text{Number of ancilla qubits}&&=\mathcal{O}\left(K\left(\log(t)+\log(\varepsilon^{-1})\right)\right).
\end{eqnarray}
Here, the Hamiltonian simulation time is the overall simulation time accounting for all $e^{-iH_it}$ dynamics. 

Here, we can consider what concrete set of jump operators can satisfy Eq. \ref{choi}. We can expand the commutator 
\begin{eqnarray}
&&[H_i\otimes H_i^*-\frac{1}{2}H_i^2\otimes I-I\otimes \frac{1}{2}H_i^{*2},H_j\otimes H_j^*-\frac{1}{2}H_j^2\otimes I-I\otimes \frac{1}{2}H_j^{*2}] \nonumber\\= && [H_i\otimes H_i^*, H_j\otimes H_j^*]   - \frac{1}{2} [H_i, H_j^2] \otimes H_i^*  - \frac{1}{2} H_i \otimes [H_i, H_j^{2}]^* \nonumber\\
&& - \frac{1}{2} [H_i^2, H_j] \otimes H_j^* - \frac{1}{2} H_j \otimes [H_i^{2}, H_j]^* \nonumber\\
&&+ \frac{1}{4} [H_i^2, H_j^2] \otimes I  + \frac{1}{4} I \otimes [H_i^{2}, H_j^{2}]^*.
\end{eqnarray}
We can find that when either $[H_i, H_j]=0$ or $\{H_i, H_j\}=0$, we will have Eq. \ref{choi} being satisfied. Therefore, we have the Corollary \ref{cccc}.
\begin{corollary}\label{cccc}
When for any $i,j$, $H_i, H_j$ are either commuting or anti-commuting, the Lindbladian Eq. \ref{clme} is a Choi commuting Lindbladian whose simulation can be fast-forwarded.
\end{corollary}

\textbf{Fast decoherence.——}A direct and surprising corollary of the fast forwarding of Choi commuting Lindbladians is that the decoherence effect under arbitrary Pauli noise can occur quadratically faster. Consider the system experiences a Pauli-type noise with the Lindbladian description having the form
\begin{eqnarray}\label{paulinoise}
\frac{d\rho(t)}{dt}=\mathcal{L}_P[\rho(t)]=\sum_{i}\lambda_i\left(P_i\rho(t) P_i-
\rho(t)\right).
\end{eqnarray}
Since two Paulis are either commuting or anti-commuting. According to Corolloary \ref{cccc}, arbitrary combinations of Pauli noises satisfy our definition of the Choi commuting Lindbladian, which indicates that the natural decoherence effect described by Eq. \ref{paulinoise} can occur in a quadratically faster time.

\section{Binomial distribution state preparation\label{sec.stateprep}}
In this section, we briefly discuss how to prepare the state
\begin{align}\label{insb}
 |S_B\rangle  = \frac{1}{\sqrt{2^{N}}}\sum_{m=0}^N\sqrt{\binom{N}{m}}|m\rangle.
\end{align}
The idea is to first show the closeness between the binary distribution and the discrete Gaussian distribution. Then, we apply the algorithm \cite{kitaev2008wavefunction} for the discrete Gaussian state preparation.

Given a Gaussian distribution $\mathcal{N}(\mu,\sigma^2)$, Ref. \cite{kitaev2008wavefunction} proposed a method for preparing the following discrete Gaussian state
\begin{align}
|\xi_{\sigma,\mu,N}\rangle = \sum_{m=0}^{N-1} \xi_{\sigma,\mu,N}(m) |m\rangle,
\end{align}
where $\xi_{\sigma,\mu,N}(m)$ is defined \cite{bauer2021practical} as
\begin{align}
\xi_{\sigma,\mu,N}^2(m)=\sum_{l=-\infty}^{\infty}\frac{1}{f(\mu,\sigma)}e^{-\frac{(m+lN-\mu)^2}{2\sigma^2}},
\end{align}
with $f(\mu,\sigma)$
\begin{eqnarray}
f(\mu,\sigma)&&=\sqrt{2\pi\sigma^2}\vartheta(\pi\mu,e^{-2\pi^2\sigma^2})\nonumber\\
&&=\sqrt{2\pi\sigma^2}\left(1+2\sum_{l=1}^\infty \cos(2\pi l\mu)e^{-2\pi^2l^2\sigma^2}\right).
\end{eqnarray}
For $\xi_{\sigma,\mu,N}(m)$, we have
\begin{eqnarray}
\left|\xi_{\sigma,\mu,N}^2(m)-\frac{1}{f(\mu,\sigma)}e^{-\frac{(m-\mu)^2}{2\sigma^2}}\right|=\sum_{l=1}^{\infty}\frac{1}{f(\mu,\sigma)}\left(e^{-\frac{(m+lN-\mu)^2}{2\sigma^2}}+e^{-\frac{(m-lN-\mu)^2}{2\sigma^2}}\right).
\end{eqnarray}
As we will see later, we are interested in $\mu=N/2$ and $\sigma^2=N/4$, which means $|m-\mu|\leq N/2$, and therefore, $|lN\pm (m-\mu)|\geq (l-1/2)N$. As a result, we have
\begin{eqnarray}\label{lllll}
\left|\xi_{\sigma,\mu,N}^2(m)-\frac{1}{f(\mu,\sigma)}e^{-\frac{(m-\mu)^2}{2\sigma^2}}\right|&&\leq 2\sum_{l=1}^{\infty}\frac{1}{f(\mu,\sigma)}e^{-\frac{(l-1/2)^2N^2}{2\sigma^2}}\nonumber\\
&&\leq 2 \frac{1}{f(\mu,\sigma)}e^{-\frac{N^2}{8\sigma^2}}+2\sum_{l=2}^{\infty}\frac{1}{f(\mu,\sigma)}e^{-\frac{lN^2}{2\sigma^2}}\nonumber\\
&&\leq 2 \frac{1}{f(\mu,\sigma)}e^{-\frac{N^2}{8\sigma^2}}+2\frac{1}{f(\mu,\sigma)}e^{-\frac{N^2}{\sigma^2}}\frac{1}{1-e^{-\frac{N^2}{2\sigma^2}}}\nonumber\\
&&=2 \frac{1}{f(\mu,\sigma)}e^{-\frac{N}{2}}+2\frac{1}{f(\mu,\sigma)}e^{-4N}\frac{1}{1-e^{-2N}}
\end{eqnarray}
where in the second line, we use the property: $(l-1/2)^2>l$ for $l\geq 2$, and in the third line, we use the sum formula for a geometric sequence. For $f(\mu,\sigma)$, we have
\begin{eqnarray}
\left|f(\mu,\sigma)-\sqrt{2\pi\sigma^2}\right|&&=\sqrt{2\pi\sigma^2}\vartheta(\pi\mu,e^{-2\pi^2\sigma^2})\nonumber\\
&&=2\sqrt{2\pi\sigma^2}\left|\sum_{l=1}^\infty \cos(2\pi l\mu)e^{-2\pi^2l^2\sigma^2}\right|\nonumber\\
&&\leq 2\sqrt{2\pi\sigma^2} \sum_{l=1}^\infty e^{-2\pi^2l\sigma^2}\nonumber\\
&&=2\sqrt{2\pi\sigma^2}e^{-2\pi^2\sigma^2}\frac{1}{1-e^{-2\pi^2\sigma^2}}\nonumber\\
&&=2\sqrt{\pi N/2}e^{-\pi^2N/2}\frac{1}{1-e^{-\pi^2N/2}}
\nonumber\\&&=\mathcal{O}(\exp(-N)),
\end{eqnarray}
putting which into Eq. \ref{lllll}, we have
\begin{eqnarray}\label{mias}
\left|\xi_{\sigma,\mu,N}^2(m)-\frac{1}{f(\mu,\sigma)}e^{-\frac{(m-\mu)^2}{2\sigma^2}}\right|=\mathcal{O}(\exp(-N)).
\end{eqnarray}
Based on Eq. \ref{mias}, we obtain the relation between $\xi_{\sigma,\mu,N}^2(m)$ and the probability distribution function (PDF) of the Gaussian distribution $\mathcal{N}(\mu,\sigma^2)$ at $\mu=N/2$ and $\sigma^2=N/4$.
\begin{eqnarray}\label{cere}
\left|\xi_{\sigma,\mu,N}^2(m)-\frac{1}{\sqrt{2\pi\sigma^2}}e^{-\frac{(m-\mu)^2}{2\sigma^2}}\right|&&\leq \left|\xi_{\sigma,\mu,N}^2(m)-\frac{1}{f(\mu,\sigma)}e^{-\frac{(m-\mu)^2}{2\sigma^2}}\right|+\left|\frac{1}{f(\mu,\sigma)}e^{-\frac{(m-\mu)^2}{2\sigma^2}}-\frac{1}{\sqrt{2\pi\sigma^2}}e^{-\frac{(m-\mu)^2}{2\sigma^2}}\right|\nonumber\\
&&=\mathcal{O}(\exp(-N)).
\end{eqnarray}

For the initial state $|S_B\rangle$ Eq. \ref{insb}, we have $\{p_m=2^{-N}\binom{N}{m}\}$ as the binomial distribution $\mathcal{B}(N,1/2)$. Regarding the relations between the binomial distribution and the PDF of $\mathcal{N}(\mu,\sigma^2)$, we have the following lemma.
\begin{lemma}[De Moivre–Laplace Theorem \cite{petrov2012sums}\label{lleemm}]
Let $S_N$ be a random variable with binomial distribution $\mathcal{B}(N,p)$. 
Let $\mu = Np$ and $\sigma^2 = Np(1-p)$. 
Let $\phi(x) = \frac{1}{\sqrt{2\pi\sigma^2}} e^{-(x-\mu)^2/(2\sigma^2)}$ be the PDF of $\mathcal{N}(\mu,\sigma^2)$. 
Then, for $m \in \{0,1,\ldots,N\}$, the following approximation holds
\begin{eqnarray}
\left|\Pr(S_N=m)-\phi(m)\right|= \mathcal{O}\left(\frac{1}{N}\right).
\end{eqnarray}
\end{lemma}
Combining Lemma \ref{lleemm} and Eq. \ref{cere}, using the triangle inequality, we obtain that 
\begin{eqnarray}
\left|p_m-\xi_{\sqrt{N}/2,N/2,N}^2(m)\right|= \mathcal{O}\left(\frac{1}{N}\right).
\end{eqnarray}

We now consider the total variance distance
\begin{eqnarray}
d_{\text{tvd}}&&=\sum_{m=0}^N\left|p_m-\xi_{\sqrt{N}/2,N/2,N}^2(m)\right|\nonumber\\
&&\leq \sum_{|m-N/2|\leq cN}\left|p_m-\xi_{\sqrt{N}/2,N/2,N}^2(m)\right| +\sum_{|m-N/2|> cN}\left(p_m+\xi_{\sqrt{N}/2,N/2,N}^2(m)\right)\nonumber\\
&&\leq cN\mathcal{O}(\frac{1}{N})+2e^{-c^2N}+2e^{-2c^2N},
\end{eqnarray}
where in the last line, we use Lemma \ref{con2} and the concentration property of the Gaussian distribution \cite{Vershynin_2018}. Let $c^2N=\log(N)$, we have
\begin{eqnarray}
d_{\text{tvd}}&&=\mathcal{O}(\frac{\log^{1/2}(N)}{\sqrt{N}})+\mathcal{O}(1/N)=\tilde{\mathcal{O}}(\frac{1}{\sqrt{N}}).
\end{eqnarray}
Using $d_{\text{tvd}}$, we can further bound the $l_2$-distance between $|S_B\rangle$ and $|\xi_{\sigma,\mu,N}\rangle$
\begin{eqnarray}
\||S_B\rangle-|\xi_{\sigma,\mu,N}\rangle\|_2&&=\sum_{m=0}^N \left|\sqrt{p_m}-\sqrt{\xi_{\sqrt{N}/2,N/2,N}(m)}\right|^2\nonumber\\&&\leq \sum_{m=0}^N \left|\left(\sqrt{p_m}-\sqrt{\xi_{\sqrt{N}/2,N/2,N}(m)}\right)\left(\sqrt{p_m}+\sqrt{\xi_{\sqrt{N}/2,N/2,N}(m)}\right)\right|\nonumber\\&&=d_{\text{tvd}}=\tilde{\mathcal{O}}(\frac{1}{\sqrt{N}}).
\end{eqnarray}
Recall that we have $N=t^3/\varepsilon^2$, therefore, we have
\begin{eqnarray}\label{surprise}
\||S_B\rangle-|\xi_{\sigma,\mu,N}\rangle\|_2=\tilde{\mathcal{O}}(\varepsilon),
\end{eqnarray}
where we ignore $t$ since we typically have $t>1$. Since $\varepsilon$ is the error of the fast-forwarding Lindbladian simulation (Section \ref{smff}), the $\tilde{\mathcal{O}}(\varepsilon)$ scaling in Eq. \ref{surprise} can make sure that replacing $|S_B\rangle$ by $|\xi_{\sigma,\mu,N}\rangle$ will not increase error.

Regarding the preparation of $|\xi_{\sigma,\mu,N}\rangle$, let $N=2^d$, the observation is that we have the recursive relation
\begin{align}
    |\xi_{\sigma,\mu,N}\rangle = |\xi_{\frac{\sigma}{2},\frac{\mu}{2},2^{d-1}}\rangle \otimes \cos\alpha|0\rangle + |\xi_{\frac{\sigma}{2},\frac{\mu-1}{2},2^{d-1}}\rangle \otimes \sin\alpha|1\rangle,
\end{align}
with
\begin{equation}
\alpha = \cos^{-1}\left(\sqrt{f\left(\frac{\mu}{2},\frac{\sigma}{2}\right) / f(\mu,\sigma)}\right),
\end{equation}
which can be computed classically. Based on this, Ref. \cite{kitaev2008wavefunction} gives a recursive quantum algorithm which can prepare $|\xi_{\sigma,\mu,N}\rangle$ within an error $\varepsilon$ using a gate complexity $\mathcal{O}(\text{poly}(n+\log(\varepsilon^{-1})))$.
One may also consider using an improved version, provided in Ref.~\cite{bagherimehrab2022nearly}.
Since the distance between $|S_B\rangle$ and $|\xi_{\sigma,\mu,N}\rangle$ is also $\tilde{\mathcal{O}}(\varepsilon)$, to make sure the fast-forwarding Lindbladian simulation algorithm in Section \ref{smff} is within error $\varepsilon$, the initial state prepation only needs a gate complexity of $\mathcal{O}(\text{poly}(n+\log(\varepsilon^{-1})))$.

\end{appendix}
\end{document}